\documentclass{eptcs}

\usepackage[english]{babel}
\usepackage{iftex}
\usepackage{graphicx}
\usepackage{mathtools}
\usepackage{cleveref}
\usepackage{todonotes}
\usepackage{amssymb,amsmath,amsthm}
\usepackage{tikzit}
\usepackage{stmaryrd}
\usepackage{soul}
\usepackage{enumitem}
\usepackage{xspace}
\usepackage[utf8]{inputenc}
\usepackage[english]{babel}
\usepackage[T1]{fontenc}
\usepackage{amsmath}
\usepackage{hyperref}
\usepackage[numbers,sort&compress]{natbib}
\usepackage{tabularx}
\usepackage{algorithm, float}
\usepackage{algpseudocode}
\definecolor{zx_grey}{RGB}{211,211,211}

\tikzstyle{Green Node}=[minimum size=5mm, font={\footnotesize\boldmath}, shape=rectangle, rounded corners=2mm, inner sep=0.2mm, outer sep=-2mm, scale=0.8, tikzit shape=circle, draw=black, fill={rgb,255: red,216; green,248; blue,216}, tikzit fill={rgb,255: red,216; green,248; blue,216}]
\tikzstyle{Empty green}=[inner sep=0mm, minimum size=2mm, shape=circle, draw=black, fill={rgb,255: red,216; green,248; blue,216}]
\tikzstyle{Red Node}=[minimum size=5mm, font={\footnotesize\boldmath}, shape=rectangle, rounded corners=2mm, inner sep=0.2mm, outer sep=-2mm, scale=0.8, tikzit shape=circle, draw=black, fill={rgb,255: red,232; green,165; blue,165}, tikzit fill={rgb,255: red,232; green,165; blue,165}]
\tikzstyle{Empty red}=[inner sep=0mm, minimum size=2mm, shape=circle, draw=black, fill={rgb,255: red,232; green,165; blue,165}]
\tikzstyle{wh}=[minimum size=5mm, font={\footnotesize\boldmath}, shape=rectangle, rounded corners=2mm, inner sep=0.2mm, outer sep=-2mm, scale=0.8, tikzit shape=circle, draw=black, fill={rgb,255: red,255; green,255; blue,255}, tikzit fill={rgb,255: red,255; green,255; blue,255}]
\tikzstyle{h}=[fill=white, draw=black, shape=rectangle, inner sep=0.6mm, minimum height=1.5mm, minimum width=1.5mm]
\tikzstyle{AutoCCZ}=[minimum size=6mm, font={\footnotesize}, fill=white, draw=black, shape=rectangle, scale=0.8]
\tikzstyle{txt}=[minimum size=5mm, font={\footnotesize\boldmath}, shape=rectangle, rounded corners=2mm, inner sep=0.2mm, outer sep=-2mm, scale=0.8, tikzit shape=circle, draw=none, fill=none, tikzit draw=none]
\tikzstyle{bigtxt}=[minimum size=5mm, font={\boldmath}, shape=rectangle, rounded corners=2mm, inner sep=0.2mm, outer sep=-2mm, scale=0.8, tikzit shape=circle, draw=none, fill=none, tikzit draw=none]
\tikzstyle{testt}=[fill=white, draw=black, shape=circle]
\tikzstyle{txtnobold}=[minimum size=5mm, font={\footnotesize}, shape=rectangle, rounded corners=2mm, inner sep=0.2mm, outer sep=-2mm, scale=0.8, tikzit shape=circle, draw=none, fill=none, tikzit draw=none]
\tikzstyle{divide}=[regular polygon, regular polygon sides=3, draw=black, fill={rgb,255: red,211; green,211; blue,211}, inner sep=2.3pt, tikzit category=scal, rounded corners=0.8mm]
\tikzstyle{black}=[fill=black, draw=black, shape=circle, tikzit fill=black, tikzit draw=black, tikzit shape=circle, tikzit category=IH, inner sep=2pt]
\tikzstyle{gather}=[fill={rgb,255: red,211; green,211; blue,211}, draw=black, tikzit category=scal, rounded corners=0.8mm, regular polygon, regular polygon sides=3, inner sep=2.3pt, rotate=180]
\tikzstyle{new style 0}=[fill=white, draw=black, shape=circle]
\tikzstyle{morphism}=[fill=white, draw=black, shape=rectangle, minimum height=5mm, font={\footnotesize}]
\tikzstyle{gnd}=[gnd]
\tikzstyle{borders}=[fill=none, draw=black, shape=circle, inner sep=-2.02mm, outer sep=0mm]
\tikzstyle{Gate control}=[fill=black, draw=black, shape=circle, tikzit fill=black, tikzit draw=black, scale=0.25]
\tikzstyle{Quantum gate}=[fill=none, draw=black, shape=rectangle, tikzit shape=rectangle, font={\footnotesize}, minimum height=5mm, scale=0.7]
\tikzstyle{smallest txt}=[fill=none, draw=none, shape=circle, font={\footnotesize}, scale=0.6]

\tikzstyle{hadamard edge}=[-, dashed, dash pattern=on 2pt off 1pt, thin, draw={rgb,255: red,68; green,136; blue,255}, fill=none]
\tikzstyle{edge}=[-, thick]
\tikzstyle{delay}=[-, fill={rgb,255: red,232; green,165; blue,165}, tikzit fill={rgb,255: red,232; green,165; blue,165}]
\tikzstyle{delay2}=[-, draw=black, fill={rgb,255: red,216; green,248; blue,216}, tikzit draw=black, tikzit fill={rgb,255: red,216; green,248; blue,216}]
\tikzstyle{arrow}=[->]
\tikzstyle{distribution}=[-, fill={rgb,255: red,116; green,179; blue,231}, tikzit fill={rgb,255: red,37; green,77; blue,255}]
\tikzstyle{discontinuous}=[-, dashed, dash pattern=on 4pt off 1pt, thin, draw=black, fill=black, tikzit fill=black, tikzit draw=black]
\tikzstyle{discon}=[-, dashed, dash pattern=on 2pt off 1pt, thin]
\tikzstyle{gzz fill}=[-, fill={rgb,255: red,216; green,248; blue,216}, dashed]

\usepackage{subfig}
\usepackage{orcidlink}

\newtheorem{definition}{Definition}

\newtheorem{proposition}{Proposition}
\theoremstyle{definition}

\makeatletter
\algnewcommand{\LineComment}[1]{\Statex \hskip\ALG@thistlm \(\triangleright\) #1}
\newcommand*{\relrelbarsep}{.386ex}
\newcommand*{\relrelbar}{%
  \mathrel{%
    \mathpalette\@relrelbar\relrelbarsep
  }%
}
\newcommand*{\@relrelbar}[2]{%
  \raise#2\hbox to 0pt{$\m@th#1\relbar$\hss}%
  \lower#2\hbox{$\m@th#1\relbar$}%
}
\providecommand*{\rightrightarrowsfill@}{%
  \arrowfill@\relrelbar\relrelbar\rightrightarrows
}
\providecommand*{\leftleftarrowsfill@}{%
  \arrowfill@\leftleftarrows\relrelbar\relrelbar
}
\providecommand*{\xrightrightarrows}[2][]{%
  \ext@arrow 0359\rightrightarrowsfill@{#1}{#2}%
}
\providecommand*{\xleftleftarrows}[2][]{%
  \ext@arrow 3095\leftleftarrowsfill@{#1}{#2}%
}

\newsavebox{\@brx}
\newcommand{\llangle}[1][]{\savebox{\@brx}{\(\m@th{#1\langle}\)}%
  \mathopen{\copy\@brx\kern-0.5\wd\@brx\usebox{\@brx}}}
\newcommand{\rrangle}[1][]{\savebox{\@brx}{\(\m@th{#1\rangle}\)}%
  \mathclose{\copy\@brx\kern-0.5\wd\@brx\usebox{\@brx}}}

\def\@seccntformat#1{\@ifundefined{#1@cntformat}%
   {\csname the#1\endcsname\quad}  
   {\csname #1@cntformat\endcsname}
}

\makeatother

\newcommand{\odd}[1]{\textsf{Odd}\left(#1\right)}
\newcommand{\cz}{\operatorname{CZ}}
\newcommand{\gms}{\operatorname{GMS}}
\newcommand{\gzz}{\operatorname{GZZ}}
\newcommand{\fo}{\operatorname{FO}}
\newcommand{\cnot}{\operatorname{CNOT}}
\newcommand{\h}{\operatorname{H}}
\newcommand{\zz}{\operatorname{ZZ}}
\newcommand{\xx}{\operatorname{XX}}
\newcommand{\z}{\operatorname{Z}}
\newcommand{\x}{\operatorname{X}}
\newcommand{\y}{\operatorname{Y}}
\newcommand{\rz}{\operatorname{R}_{\z}}
\newcommand{\rx}{\operatorname{R}_{\x}}
\newcommand{\ry}{\operatorname{R}_{\y}}
\newcommand{\lp}{LP\xspace}
\newcommand{\rr}{\operatorname{R}}

\ifpdf
  \usepackage{underscore}         
  \usepackage[T1]{fontenc}        
\else
  \usepackage{breakurl}           
\fi

\newcommand{\comp}[1]{\overline{#1}}

\title{Optimization and Synthesis of Quantum Circuits\\with Global Gates}

\author{
  Alejandro Villoria\orcidlink{0000-0003-0787-2568}
  \qquad Henning Basold\orcidlink{0000-0001-7610-8331}
  \qquad Alfons Laarman\orcidlink{0000-0002-2433-4174}
  \institute{Leiden Institute of Advanced Computer Science, The Netherlands}
  \email{ \{a.d.villoria.gonzalez, h.basold, a.w.laarman\}@liacs.leidenuniv.nl } }

\def\titlerunning{Optimization and Synthesis of Quantum Circuits with Global Gates}
\def\authorrunning{A. Villoria, H. Basold, A. Laarman}
\begin{document}
\maketitle
\begin{abstract}
  Compiling quantum circuits to account for hardware restrictions
  is an essential part of the quantum computing stack.
  Circuit compilation allows us to adapt algorithm descriptions into a sequence of operations
  supported by real quantum hardware, and has the potential to significantly
  improve their performance when optimization techniques are added to the process.
  One such optimization technique is reducing the number of quantum gates that are
  needed to execute a circuit.
  For instance, methods for reducing the number of non-Clifford gates or CNOT gates from
  a circuit are an extensive research area that has gathered significant interest over the years.
  For certain hardware platforms such as ion trap quantum computers, we can leverage some of their
  special properties to further reduce the cost of executing a quantum circuit in them.
  In this work we use global interactions, such as the Global M{\o}lmer-S{\o}rensen
  gate present in ion trap hardware, to optimize and synthesize quantum circuits.
  We design and implement an algorithm that is able to compile an arbitrary quantum circuit into
  another circuit that uses global gates as the entangling operation, while optimizing the number
  of global interactions needed.
  The algorithm is based on the ZX-calculus and uses a specialized circuit extraction
  routine that groups entangling gates into Global M{\o}lmer-S{\o}rensen gates.
  We benchmark the algorithm in a variety of circuits, and show how it improves their performance
  under state-of-the-art hardware considerations in comparison to a naive algorithm and
  the Qiskit optimizer.
\end{abstract}
\section{Introduction}

Physical realizations of quantum computers in the current era vary in multiple,
non-trivial aspects.
A fundamental difference between two hardware realizations of a quantum computer can be
the choice of implementation of their qubits.
Current candidates for qubit implementations are superconducting qubits~\cite{kim-evidence-2023},
trapped ions~\cite{figgatt-parallel-2019}, nitrogen-vacancy
centers~\cite{gulka-room-temperature-2021}, and neutral atoms~\cite{henriet-quantum-2020},
to name a few.
Each choice of qubit implementation comes with distinct features that need to be considered
when compiling and optimizing circuits for a quantum computer based on them.
For instance, the connectivity between qubits in a superconducting quantum computer
is limited to a specific topology, while
on a trapped ion quantum computer the connectivity is all-to-all.
This all-to-all connectivity naturally gives us a special set of quantum gates to work
with, in which we have access to powerful multi-qubit entangling gates that we refer to as
\emph{global gates}~\cite{figgatt-parallel-2019}.
Thus, if we wish to run an arbitrary quantum circuit on a trapped ion quantum computer,
we need a way to synthesize the circuit such that we take advantage of the global gates
we have available.

Circuit compilation and optimization are parts of the quantum computing stack that are
instrumental to enable useful computations in real quantum devices, which is especially relevant
for Noisy Intermediate-Scale Quantum (NISQ) era devices due to their limited resources
~\cite{Corcoles-2020}.
In particular, the number of entangling operations, rather than single-qubit ones,
is usually the target to be reduced in works concerned with current-era devices.
In the case of ion trap quantum computers, we have that both global interactions and
two-qubit gates such as the M{\o}lmer-S{\o}rensen $\xx$ gate are considerably more expensive
in terms of execution time and fidelity than single-qubit operations
\cite{bermudez-assessing-2017}.
Given that global gates such as the GMS gate are in fact a series of two-qubit $\xx$ gates
executed simultaneously, we can see how grouping $\xx$ gates into a low number of GMS
gates can be beneficial for the overall performance of the circuit.

The process of optimizing and synthesizing quantum circuits greatly varies depending on the tools
used for the task.
One such tool that has gained traction in the past years is the ZX-calculus
\cite{coecke-interacting-2011}.
The ZX-calculus is a graphical language that can be used to represent
quantum circuits as diagrams and to perform rewrites on them.
Rewriting circuits with the ZX-calculus can lead to algorithms
that reduce the gate counts of the circuits they represent~\cite{duncan-graph-theoretic-2020}.
Such rewriting algorithms consist of first reducing the size of the diagram and then
\emph{extracting} a (optimized) quantum circuit from it.
Being the final step in the process, we can think of circuit extraction as a way of synthesizing
our circuit.
Indeed, we can adapt the extraction algorithm to
ensure that we output a quantum circuit that accommodates the specific requirements of the hardware
platform we have in mind.
In this work we present a circuit extraction algorithm that extracts entangling gates as
global gates while keeping the total amount of global gates low, which in turn yields a quantum
circuit compiled for ion trap hardware.
By leveraging the ZX-calculus, we allow for the compilation of arbitrary quantum circuits as
opposed to most works on compilation with global gates,
where only subclasses of quantum circuits are studied.

\paragraph{Related work}\mbox{}\\
There have been multiple works on the topic of circuit compilation for ion trap
architectures, with a focus on reducing the amount of GMS gates and studying how they affect the
overall circuit performance.
Some works have focused on compiling
Clifford circuits using global gates,
with results that started with implementations using a number of global gates that
went from scaling linearly
with the number of qubits~\cite{maslov-use-2018, wetering-constructing-2021,
  grzesiak-efficient-2022, basler-synthesis-2023},
to a constant amount~\cite{bravyi-constant-cost-2022}.
There has also been work on compiling particular (not necessarily Clifford) circuits
of interest with global gates,
such as Toffoli-$n$, QFT, QRAM, Quantum Volume circuits, and Diagonal unitaries
\cite{maslov-use-2018, basler-synthesis-2023, nemirovsky-efficient-2025,
  allcock-constant-depth-2024}.
One work has tackled the task of compiling arbitrary circuits with global gates, where an upper
bound on the maximum number of GMS gates needed to compile the circuit is given,
which depends on the qubit count and number of non-Clifford gates, though no benchmarks for this
method are known~\cite{wetering-constructing-2021}.
On a similar line, other works have studied the performance of GMS gates given different
constraints in the behaviour of those gates and on a variety of other hardware characteristics,
such as different gate fidelities and limited connectivity.
It was shown that GMS gates of limited capability (only targeting four qubits at a time)
need to have a fidelity of $98\%$-$99\%$ in order to surpass implementations with two-qubit gates
on systems of size up to $20$ qubits~\cite{kumar-digital-analog-2025}.
The performance of five-qubit GMS
gates against regular two-qubit gates for error correction applications has also been
studied~\cite{bermudez-assessing-2017}, and it was shown that
having GMS gates heavily reduces the number of other necessary physical operations such as ion
shuttling and swapping.
Lastly, it was demonstrated that in some specific architectures (such as star-shaped),
GMS gates that affect the whole qubit register perform better than regular two-qubit gates
\cite{canelles2023benchmarking}.

\paragraph{Contributions}\mbox{}\\
We propose a method for converting an arbitrary quantum circuit into one using the gate-set
native to ion trap quantum platforms.
Our synthesis algorithm is based on the ZX-calculus and consists of a circuit extraction procedure
that outputs a circuit using GMS gates and single-qubit rotations.
Our algorithm can be split into two parts.
On one side, during circuit extraction, we constrain how the layers of CNOT gates of the
output circuit are shaped so they are more amenable to be compiled with a single GMS.
Secondly, we add a variety of peephole optimization style rewrites during the extraction process
to group entangling gates into GMS gates and to reduce the amount of single-qubit gates.
In both techniques, we aim to minimize the number of global gates we generate,
which are powerful but more costly than single-qubit operations.
We develop a Python implementation of our algorithm and run benchmarks over different
classes of quantum circuits.

\section{Ion Trap Quantum Hardware}\label{sec:ions}
Quantum computing with trapped ions is an approach to quantum information processing in
which a linear chain of atomic ions, suspended by an electric field, are operated on via
focused laser beams to perform quantum logic gates~\cite{Chen2024benchmarkingtrapped}. 
The set of operations arising from this type of
hardware consists of single-qubit rotations
and an entangling two-qubit $\xx$ gate, sometimes referred to as the M{\o}lmer-S{\o}rensen
(MS) or Ising gate.
These gates are defined as follows:
\begin{equation}
  \label{eq:nativegates}
  \rr(\theta,\phi) = e^{-i(\theta/2)(\cos\phi\x + \sin\phi\y)},
  \quad\quad \rz(\theta) = e^{-i(\theta/2)\z},
  \quad\quad \xx(\alpha)_{i,j} = e^{-i\frac{\alpha}{2}\x_i\x_j}.
\end{equation}
Where $\rr(\theta,\phi)$ is a rotation by an angle of $\theta$ with phase $\phi$,
$\x,\y$, and $\z$ are the Pauli gates, $\alpha$ is an angle in $ [0,2\pi)$,
and $\x_i\x_j$ are the Pauli $\x$ on qubits $i$ and $j$, respectively.
We can recover the other two better-known Pauli rotation gates by noticing
$\rx(\theta)=\rr(\theta,0)$ and $\ry(\theta)=\rr(\theta,\pi/2)$.
Given this straightforward conversion, in this work we mostly reason with the Pauli rotation gates.

On top of this (universal) set of operations, ion trap hardware is naturally capable
of global entangling $\xx$ interactions
(referred to as Global MS gates, or GMS) due to its all-to-all connectivity.
GMS gates consist of a sequence of two-qubit $\xx$ gates performed in parallel
between a subset of the available qubits, meaning that it can range from being a two-qubit
gate to affecting the whole qubit register.
We define the GMS gate for an angle $\alpha$ and a symmetric binary matrix $A$ with
zeroes in the diagonal where $A_{i,j} = 1$ if there is an $\xx$ interaction between
qubits $i$ and $j$:
\begin{equation}
  \label{eq:gmsdef}
  \gms_A(\alpha) = \prod_{i> j\in A}\xx(\alpha A_{i,j})_{i,j}.
\end{equation}
Intuitively, this gate consists of a composition of commuting two-qubit phase gates.
We can also define similar two-qubit and global interactions on the $\z$ axis by placing
Hadamard gates before and after applying an $\xx$ or a GMS on the involved qubits.
With this we get the definitions
$\zz(\alpha)_{i,j} := \h_i\h_j\ \xx(\alpha)_{i,j}\ \h_i\h_j$ and
$\gzz_A(\alpha) := \h^{\otimes n}\ \gms_A(\alpha)\  \h^{\otimes n} =  \prod_{i> j\in A}\zz(\alpha A_{i,j})_{i,j}$
where we write $\h^{\otimes n}$ for the $n$-th fold tensor product of the Hadamard gate.

Due to the various ways it can be implemented in real hardware, we can find in the literature
variations in the definition of the GMS gate~\cite{maslov-use-2018}.
Other, more restrictive definitions of the GMS have to involve all of the qubits in the register,
disallowing the targeting of a specific subset of the qubits.
In this case we can always manually exclude qubits by applying an additional sequence of 
these restricted GMS gates and some single-qubit gates, the amount of which
scales linearly with the qubit count~\cite{wetering-constructing-2021}.
Similarly, the efficient, arbitrary, simultaneously entangling (EASE) gate allows for
different couplings $\alpha_{i,j}$
for each pair of qubits $i,j$
at the cost of added
complexity on the pulse design
\cite{grzesiak-efficient-2022, grzesiak-efficient-2020}.

In this work we assume that the ion trap device is natively capable of executing $(1)$ the
$\rr(\theta,\phi)$ and $\rz(\theta)$ gates on any qubit $(2)$ any single-qubit gate
in parallel $(3)$ $\xx(\alpha)$ on any qubit pair and $(4)$ the GMS of~\Cref{eq:gmsdef} with the
ability of excluding certain qubits
(but we do not require different couplings for each pair of qubits).

\section{ZX-calculus}\label{sec:zx}

We introduce in this section the ZX-calculus, a formal diagrammatic language
that is used for representing quantum
computations~\cite{coecke-interacting-2011}.
We present its basic definitions and properties, how quantum circuits can be
represented as ZX-diagrams, and how the process of diagram simplification and circuit
extraction from a ZX-diagram works.

\subsection{Definition}
The ZX-calculus consists of a set of graphical generators (referred to as \textit{ZX-diagrams})
that are interpreted as linear maps
$L: \mathbb{C}^{2^{n}}\to \mathbb{C}^{2^{m}}$
in Hilbert spaces in the following way, where $\alpha\in\mathbb{R}$
and $|\psi^{(n)}\rangle$ is the shorthand for the $n$-fold tensor product of $|\psi\rangle$:
\begin{equation*}
  \scalebox{0.99}{\begin{tikzpicture}
	\begin{pgfonlayer}{nodelayer}
		\node [style=none] (198) at (11, -0.25) {};
		\node [style=none] (199) at (11, 1.25) {};
		\node [style=none] (200) at (-10.5, 5) {};
		\node [style=none] (201) at (-10.5, 3) {};
		\node [style=none] (202) at (-10.5, 4.5) {};
		\node [style=none] (203) at (-9.975, 4) {$\vdots$};
		\node [style=none] (204) at (-7.5, 5) {};
		\node [style=none] (205) at (-7.5, 3) {};
		\node [style=none] (206) at (-7.5, 4.5) {};
		\node [style=none] (207) at (-8.025, 4) {$\vdots$};
		\node [style=Green Node] (208) at (-9, 4) {};
		\node [style=Green Node] (211) at (-9, 4) {$\alpha$};
		\node [style=none] (212) at (1.5, 5) {};
		\node [style=none] (213) at (1.5, 3) {};
		\node [style=none] (214) at (1.5, 4.5) {};
		\node [style=none] (215) at (2.025, 4) {$\vdots$};
		\node [style=none] (216) at (4.5, 5) {};
		\node [style=none] (217) at (4.5, 3) {};
		\node [style=none] (218) at (4.5, 4.5) {};
		\node [style=none] (219) at (3.975, 4) {$\vdots$};
		\node [style=Red Node] (220) at (3, 4) {};
		\node [style=Red Node] (221) at (3, 4) {$\alpha$};
		\node [style=none] (222) at (14.25, 4) {};
		\node [style=none] (223) at (16.75, 4) {};
		\node [style=h] (224) at (15.5, 4) {};
		\node [style=none] (225) at (6.5, -0.25) {};
		\node [style=none] (226) at (6.5, 1.25) {};
		\node [style=none] (227) at (-2.25, 0.5) {};
		\node [style=none] (228) at (0.25, 0.5) {};
		\node [style=none] (229) at (-8, 1) {};
		\node [style=none] (230) at (-10, 1) {};
		\node [style=none] (231) at (-8, -0.25) {};
		\node [style=none] (232) at (-10, -0.25) {};
		\node [style=none] (242) at (18.5, 1) {};
		\node [style=none] (243) at (19.5, 1) {};
		\node [style=none] (244) at (18.5, 0) {};
		\node [style=none] (245) at (19.5, 0) {};
		\node [style=txtnobold] (247) at (-7.5, 4) {$m$};
		\node [style=txtnobold] (249) at (1.5, 4) {$n$};
		\node [style=txtnobold] (251) at (-10.5, 4) {$n$};
		\node [style=txtnobold] (253) at (4.5, 4) {$m$};
		\node [style=txtnobold] (263) at (-3.5, 4) {$:=\ \ |0^{(m)}\rangle\langle0^{(n)}| + e^{i\alpha}|1^{(m)}\rangle\langle 1^{(n)}|$};
		\node [style=txtnobold] (264) at (9, 4) {$:=\ \ |+^{(m)}\rangle\langle +^{(n)}| + e^{i\alpha}|-^{(m)}\rangle\langle -^{(n)}|$};
		\node [style=txtnobold] (265) at (19.25, 4) {$:=\ \ \frac{1}{\sqrt{2}}\begin{bmatrix}1 & 1\\ 1 & -1 \end{bmatrix}$};
		\node [style=txtnobold] (266) at (2.25, 0.5) {$:=\ \ \begin{bmatrix}1 & 0\\ 0 & 1 \end{bmatrix}$};
		\node [style=txtnobold] (267) at (-6, 0.5) {$:=\ \ \begin{bmatrix}1 & 0 & 0 & 0\\ 0 & 0 & 1& 0\\  0 & 1 & 0& 0 \\  0 & 0 & 0& 1  \end{bmatrix}$};
		\node [style=txtnobold] (268) at (14.75, 0.5) {$:=\ \ \begin{bmatrix} 1 & 0 & 0 & 1  \end{bmatrix}$};
		\node [style=txtnobold] (269) at (7.75, 0.5) {$:=\ \ \begin{bmatrix} 1 \\ 0 \\ 0 \\ 1  \end{bmatrix}$};
		\node [style=txtnobold] (270) at (20.75, 0.5) {$:=\ \ 1$};
	\end{pgfonlayer}
	\begin{pgfonlayer}{edgelayer}
		\draw [style=edge, bend right=105, looseness=2.25] (198.center) to (199.center);
		\draw [style=edge, bend right] (208) to (200.center);
		\draw [style=edge, bend right=15] (208) to (202.center);
		\draw [style=edge, bend left=15] (208) to (201.center);
		\draw [style=edge, bend left] (208) to (204.center);
		\draw [style=edge, bend left=15] (208) to (206.center);
		\draw [style=edge, bend right=15] (208) to (205.center);
		\draw [style=edge, bend right] (220) to (212.center);
		\draw [style=edge, bend right=15] (220) to (214.center);
		\draw [style=edge, bend left=15] (220) to (213.center);
		\draw [style=edge, bend left] (220) to (216.center);
		\draw [style=edge, bend left=15] (220) to (218.center);
		\draw [style=edge, bend right=15] (220) to (217.center);
		\draw [style=edge] (222.center) to (224);
		\draw [style=edge] (224) to (223.center);
		\draw [style=edge, bend left=105, looseness=2.25] (225.center) to (226.center);
		\draw [style=edge] (227.center) to (228.center);
		\draw [style=edge, in=180, out=0] (230.center) to (231.center);
		\draw [style=edge, in=180, out=0] (232.center) to (229.center);
		\draw [style=discontinuous] (242.center) to (243.center);
		\draw [style=discontinuous] (243.center) to (245.center);
		\draw [style=discontinuous] (245.center) to (244.center);
		\draw [style=discontinuous] (242.center) to (244.center);
	\end{pgfonlayer}
\end{tikzpicture}}
\end{equation*}
From left to right, the generators of the first row are referred to as
the $Z$- and $X$-\textit{spiders} (or green and red spiders, respectively) and the
Hadamard gate.
On the second row we have the swap, identity wire, \textit{cap}, \textit{cup}, and empty diagram.
Similar to quantum circuit notation, we will be reading diagrams from left (inputs) to
right (outputs), where sequential
composition is shown as connecting output and input wires, and parallel composition as
vertical arrangement of diagrams.
Additionally, we oftentimes represent the zero-phase spiders as empty and the Hadamard gate as
a \emph{Hadamard edge}:
\begin{equation*}
  \label{eq:syntacticSugar}
  \scalebox{1.2}{\begin{tikzpicture}
	\begin{pgfonlayer}{nodelayer}
		\node [style=none] (1) at (5.25, 0) {};
		\node [style=none] (2) at (7.25, 0) {};
		\node [style=h] (3) at (6.25, 0) {};
		\node [style=txtnobold] (4) at (7.75, -0.08) {$=$};
		\node [style=none] (5) at (10.25, 0) {};
		\node [style=none] (6) at (8.25, 0) {};
		\node [style=none] (7) at (-9.75, 1) {};
		\node [style=none] (8) at (-9.75, -1) {};
		\node [style=none] (9) at (-9.75, 0.5) {};
		\node [style=none] (10) at (-10.275, 0) {$\vdots$};
		\node [style=none] (11) at (-12.75, 1) {};
		\node [style=none] (12) at (-12.75, -1) {};
		\node [style=none] (13) at (-12.75, 0.5) {};
		\node [style=none] (14) at (-12.225, 0) {$\vdots$};
		\node [style=Green Node] (15) at (-11.25, 0) {};
		\node [style=Green Node] (16) at (-11.25, 0) {$0$};
		\node [style=txtnobold] (17) at (-9, -0.08) {$=$};
		\node [style=none] (19) at (-5.25, 1) {};
		\node [style=none] (20) at (-5.25, -1) {};
		\node [style=none] (21) at (-5.25, 0.5) {};
		\node [style=none] (22) at (-5.775, 0) {$\vdots$};
		\node [style=none] (23) at (-8.25, 1) {};
		\node [style=none] (24) at (-8.25, -1) {};
		\node [style=none] (25) at (-8.25, 0.5) {};
		\node [style=none] (26) at (-7.725, 0) {$\vdots$};
		\node [style=Empty green] (27) at (-6.75, 0) {};
		\node [style=none] (28) at (-0.75, 1) {};
		\node [style=none] (29) at (-0.75, -1) {};
		\node [style=none] (30) at (-0.75, 0.5) {};
		\node [style=none] (31) at (-1.275, 0) {$\vdots$};
		\node [style=none] (32) at (-3.75, 1) {};
		\node [style=none] (33) at (-3.75, -1) {};
		\node [style=none] (34) at (-3.75, 0.5) {};
		\node [style=none] (35) at (-3.225, 0) {$\vdots$};
		\node [style=Green Node] (36) at (-2.25, 0) {};
		\node [style=Red Node] (37) at (-2.25, 0) {$0$};
		\node [style=txtnobold] (38) at (0, -0.08) {$=$};
		\node [style=none] (39) at (3.75, 1) {};
		\node [style=none] (40) at (3.75, -1) {};
		\node [style=none] (41) at (3.75, 0.5) {};
		\node [style=none] (42) at (3.225, 0) {$\vdots$};
		\node [style=none] (43) at (0.75, 1) {};
		\node [style=none] (44) at (0.75, -1) {};
		\node [style=none] (45) at (0.75, 0.5) {};
		\node [style=none] (46) at (1.275, 0) {$\vdots$};
		\node [style=Empty red] (47) at (2.25, 0) {};
	\end{pgfonlayer}
	\begin{pgfonlayer}{edgelayer}
		\draw [style=edge] (1.center) to (3);
		\draw [style=edge] (2.center) to (3);
		\draw [style=hadamard edge] (6.center) to (5.center);
		\draw [style=edge, bend left] (15) to (7.center);
		\draw [style=edge, bend left=15] (15) to (9.center);
		\draw [style=edge, bend right=15] (15) to (8.center);
		\draw [style=edge, bend right] (15) to (11.center);
		\draw [style=edge, bend right=15] (15) to (13.center);
		\draw [style=edge, bend left=15] (15) to (12.center);
		\draw [style=edge, bend left] (27) to (19.center);
		\draw [style=edge, bend left=15] (27) to (21.center);
		\draw [style=edge, bend right=15] (27) to (20.center);
		\draw [style=edge, bend right] (27) to (23.center);
		\draw [style=edge, bend right=15] (27) to (25.center);
		\draw [style=edge, bend left=15] (27) to (24.center);
		\draw [style=edge, bend left] (36) to (28.center);
		\draw [style=edge, bend left=15] (36) to (30.center);
		\draw [style=edge, bend right=15] (36) to (29.center);
		\draw [style=edge, bend right] (36) to (32.center);
		\draw [style=edge, bend right=15] (36) to (34.center);
		\draw [style=edge, bend left=15] (36) to (33.center);
		\draw [style=edge, bend left] (47) to (39.center);
		\draw [style=edge, bend left=15] (47) to (41.center);
		\draw [style=edge, bend right=15] (47) to (40.center);
		\draw [style=edge, bend right] (47) to (43.center);
		\draw [style=edge, bend right=15] (47) to (45.center);
		\draw [style=edge, bend left=15] (47) to (44.center);
	\end{pgfonlayer}
\end{tikzpicture}}.
\end{equation*}
Some relevant quantum gates have simple representations as ZX-diagrams:
\begin{equation}
  \label{eq:gateExamples}
  \scalebox{1.2}{\begin{tikzpicture}
	\begin{pgfonlayer}{nodelayer}
		\node [style=Green Node] (0) at (3.5, 1.5) {$\alpha$};
		\node [style=Red Node] (1) at (-1.75, 1.5) {$\alpha$};
		\node [style=none] (2) at (-2.75, 1.5) {};
		\node [style=none] (3) at (-0.75, 1.5) {};
		\node [style=none] (4) at (2.5, 1.5) {};
		\node [style=none] (5) at (4.5, 1.5) {};
		\node [style=Empty green] (6) at (-6.75, 0) {};
		\node [style=Empty red] (7) at (-6.75, -1.5) {};
		\node [style=none] (8) at (-5.75, 0) {};
		\node [style=none] (9) at (-7.75, 0) {};
		\node [style=none] (10) at (-7.75, -1.5) {};
		\node [style=none] (11) at (-5.75, -1.5) {};
		\node [style=Empty green] (12) at (-2.25, 0) {};
		\node [style=Empty green] (13) at (-2.25, -1.5) {};
		\node [style=h] (14) at (-2.25, -0.75) {};
		\node [style=none] (15) at (-3.25, 0) {};
		\node [style=none] (16) at (-1.25, 0) {};
		\node [style=none] (17) at (-3.25, -1.5) {};
		\node [style=none] (18) at (-1.25, -1.5) {};
		\node [style=Empty red] (19) at (3, 0) {};
		\node [style=Empty red] (20) at (3, -1.5) {};
		\node [style=Empty red] (21) at (8.75, -0.75) {};
		\node [style=Red Node] (22) at (4, -0.75) {$\alpha$};
		\node [style=Green Node] (23) at (9.75, -0.75) {$\alpha$};
		\node [style=Empty green] (24) at (8.75, 0) {};
		\node [style=Empty green] (25) at (8.75, -1.5) {};
		\node [style=Empty green] (26) at (3, -0.75) {};
		\node [style=none] (27) at (2, 0) {};
		\node [style=none] (28) at (4.75, 0) {};
		\node [style=none] (29) at (2, -1.5) {};
		\node [style=none] (30) at (4.75, -1.5) {};
		\node [style=none] (31) at (7.75, 0) {};
		\node [style=none] (32) at (10.5, 0) {};
		\node [style=none] (33) at (7.75, -1.5) {};
		\node [style=none] (34) at (10.5, -1.5) {};
		\node [style=txtnobold] (35) at (-4, 1.5) {$\rx(\alpha)=$};
		\node [style=txtnobold] (36) at (1.25, 1.5) {$\rz(\alpha)=$};
		\node [style=txtnobold] (37) at (-9.25, -0.5) {CNOT $=$};
		\node [style=txtnobold] (38) at (-4, -0.5) {CZ $=$};
		\node [style=txtnobold] (39) at (1, -0.5) {$\xx(\alpha)=$};
		\node [style=txtnobold] (40) at (6.75, -0.5) {$\zz(\alpha)=$};
	\end{pgfonlayer}
	\begin{pgfonlayer}{edgelayer}
		\draw [style=edge] (2.center) to (1);
		\draw [style=edge] (1) to (3.center);
		\draw [style=edge] (4.center) to (0);
		\draw [style=edge] (0) to (5.center);
		\draw [style=edge] (9.center) to (6);
		\draw [style=edge] (6) to (8.center);
		\draw [style=edge] (10.center) to (7);
		\draw [style=edge] (7) to (11.center);
		\draw [style=edge] (6) to (7);
		\draw [style=edge] (15.center) to (12);
		\draw [style=edge] (12) to (16.center);
		\draw [style=edge] (12) to (14);
		\draw [style=edge] (14) to (13);
		\draw [style=edge] (17.center) to (13);
		\draw [style=edge] (13) to (18.center);
		\draw [style=edge] (19) to (28.center);
		\draw [style=edge] (27.center) to (19);
		\draw [style=edge] (19) to (26);
		\draw [style=edge] (26) to (22);
		\draw [style=edge] (26) to (20);
		\draw [style=edge] (20) to (29.center);
		\draw [style=edge] (20) to (30.center);
		\draw [style=edge] (31.center) to (24);
		\draw [style=edge] (24) to (32.center);
		\draw [style=edge] (24) to (21);
		\draw [style=edge] (21) to (23);
		\draw [style=edge] (21) to (25);
		\draw [style=edge] (33.center) to (25);
		\draw [style=edge] (25) to (34.center);
	\end{pgfonlayer}
\end{tikzpicture}}
\end{equation}
Apart from the graphical generators, the ZX-calculus comes equipped with a
series of \textit{rewrite rules} that allows
us to modify ZX-diagrams, while preserving their underlying linear map
\cite{coecke-interacting-2011, wetering-zx-calculus-2020}.
We show some of the rules in~\Cref{fig:ruleset}.
%

\begin{figure}
  \centering
  \scalebox{1.2}{%
    \begin{tikzpicture}
	\begin{pgfonlayer}{nodelayer}
		\node [style=none] (135) at (-0.75, 1.5) {};
		\node [style=none] (136) at (-0.75, -0.5) {};
		\node [style=none] (137) at (-0.75, 1) {};
		\node [style=none] (138) at (-1.275, 0.5) {$\vdots$};
		\node [style=none] (139) at (-3.75, 1.5) {};
		\node [style=none] (140) at (-3.75, -0.5) {};
		\node [style=none] (141) at (-3.75, 1) {};
		\node [style=none] (142) at (-3.225, 0.5) {$\vdots$};
		\node [style=Green Node] (143) at (-2.25, 0.5) {};
		\node [style=Green Node] (144) at (-2.25, 0.5) {$\alpha+\beta$};
		\node [style=txtnobold] (145) at (-5, 0.5) {$\stackrel{(f)}{=}$};
		\node [style=none] (147) at (4.75, -3.95) {};
		\node [style=none] (148) at (2.25, -3.95) {};
		\node [style=txtnobold] (149) at (5.5, -3.7) {$\stackrel{(id)}{=}$};
		\node [style=none] (150) at (9, -3.95) {};
		\node [style=none] (151) at (6.5, -3.95) {};
		\node [style=Empty green] (177) at (3.5, -3.95) {};
		\node [style=none] (195) at (4.75, 1.25) {};
		\node [style=none] (196) at (4.75, -0.75) {};
		\node [style=none] (197) at (4.75, 0.75) {};
		\node [style=none] (198) at (4.225, 0.25) {$\vdots$};
		\node [style=none] (199) at (1.75, 1.25) {};
		\node [style=none] (200) at (1.75, -0.75) {};
		\node [style=none] (201) at (1.75, 0.75) {};
		\node [style=none] (202) at (2.275, 0.25) {$\vdots$};
		\node [style=Green Node] (203) at (3.25, 0.25) {};
		\node [style=Green Node] (204) at (3.25, 0.25) {$\alpha$};
		\node [style=h] (205) at (4.25, 1.25) {};
		\node [style=h] (206) at (4.25, 0.75) {};
		\node [style=h] (207) at (4.25, -0.75) {};
		\node [style=h] (208) at (2.25, 1.25) {};
		\node [style=h] (209) at (2.25, -0.75) {};
		\node [style=h] (210) at (2.25, 0.75) {};
		\node [style=none] (211) at (9.25, 1.25) {};
		\node [style=none] (212) at (9.25, -0.75) {};
		\node [style=none] (213) at (9.25, 0.75) {};
		\node [style=none] (214) at (8.725, 0.25) {$\vdots$};
		\node [style=none] (215) at (6.25, 1.25) {};
		\node [style=none] (216) at (6.25, -0.75) {};
		\node [style=none] (217) at (6.25, 0.75) {};
		\node [style=none] (218) at (6.775, 0.25) {$\vdots$};
		\node [style=Green Node] (219) at (7.75, 0.25) {};
		\node [style=Red Node] (220) at (7.75, 0.25) {$\alpha$};
		\node [style=txtnobold] (221) at (5.5, 0.5) {$\stackrel{(h)}{=}$};
		\node [style=txtnobold] (224) at (-5, -3.7) {$\stackrel{(h2)}{=}$};
		\node [style=h] (227) at (-8, -3.75) {};
		\node [style=h] (228) at (-7, -3.75) {};
		\node [style=none] (229) at (-1, -3.75) {};
		\node [style=none] (230) at (-4, -3.75) {};
		\node [style=none] (231) at (-6, -3.75) {};
		\node [style=none] (232) at (-9, -3.75) {};
		\node [style=none] (233) at (-9.75, 3) {};
		\node [style=none] (234) at (-9.75, 1) {};
		\node [style=none] (235) at (-9.75, 2.5) {};
		\node [style=none] (236) at (-9.225, 2) {$\vdots$};
		\node [style=none] (237) at (-6.75, 3) {};
		\node [style=none] (238) at (-6.75, 1) {};
		\node [style=none] (239) at (-6.75, 2.5) {};
		\node [style=none] (240) at (-7.275, 2) {$\vdots$};
		\node [style=Green Node] (241) at (-8.25, 2) {};
		\node [style=Green Node] (242) at (-8.25, 2) {$\alpha$};
		\node [style=none] (243) at (-9, 0) {};
		\node [style=none] (244) at (-9, -2) {};
		\node [style=none] (245) at (-9, -0.5) {};
		\node [style=none] (246) at (-8.475, -1) {$\vdots$};
		\node [style=none] (247) at (-6, 0) {};
		\node [style=none] (248) at (-6, -2) {};
		\node [style=none] (249) at (-6, -0.5) {};
		\node [style=none] (250) at (-6.525, -1) {$\vdots$};
		\node [style=Green Node] (251) at (-7.5, -1) {};
		\node [style=Green Node] (252) at (-7.5, -1) {$\beta$};
		\node [style=none] (253) at (-7.775, 0.75) {$\vdots$};
		\node [style=none] (254) at (-7.5, -0.75) {};
	\end{pgfonlayer}
	\begin{pgfonlayer}{edgelayer}
		\draw [style=edge, bend left] (143) to (135.center);
		\draw [style=edge, bend left=15] (143) to (137.center);
		\draw [style=edge, bend right=15] (143) to (136.center);
		\draw [style=edge, bend right] (143) to (139.center);
		\draw [style=edge, bend right=15] (143) to (141.center);
		\draw [style=edge, bend left=15] (143) to (140.center);
		\draw [style=edge] (150.center) to (151.center);
		\draw [style=edge] (147.center) to (177);
		\draw [style=edge] (177) to (148.center);
		\draw [style=edge, bend right] (205) to (204);
		\draw [style=edge] (205) to (195.center);
		\draw [style=edge] (206) to (197.center);
		\draw [style=edge, bend right=15] (206) to (204);
		\draw [style=edge, bend right] (204) to (207);
		\draw [style=edge] (207) to (196.center);
		\draw [style=edge, bend right] (204) to (208);
		\draw [style=edge] (208) to (199.center);
		\draw [style=edge] (210) to (201.center);
		\draw [style=edge, bend left] (204) to (209);
		\draw [style=edge] (209) to (200.center);
		\draw [style=edge, bend right=15] (204) to (210);
		\draw [style=edge, bend left] (219) to (211.center);
		\draw [style=edge, bend left=15] (219) to (213.center);
		\draw [style=edge, bend right=15] (219) to (212.center);
		\draw [style=edge, bend right] (219) to (215.center);
		\draw [style=edge, bend right=15] (219) to (217.center);
		\draw [style=edge, bend left=15] (219) to (216.center);
		\draw [style=edge] (227) to (228);
		\draw [style=edge] (230.center) to (229.center);
		\draw [style=edge] (232.center) to (227);
		\draw [style=edge] (228) to (231.center);
		\draw [style=edge, bend right] (241) to (233.center);
		\draw [style=edge, bend right=15] (241) to (235.center);
		\draw [style=edge, bend left=15] (241) to (234.center);
		\draw [style=edge, bend left] (241) to (237.center);
		\draw [style=edge, bend left=15] (241) to (239.center);
		\draw [style=edge, bend right=15] (241) to (238.center);
		\draw [style=edge, bend right] (251) to (243.center);
		\draw [style=edge, bend right=15] (251) to (245.center);
		\draw [style=edge, bend left=15] (251) to (244.center);
		\draw [style=edge, bend left] (251) to (247.center);
		\draw [style=edge, bend left=15] (251) to (249.center);
		\draw [style=edge, bend right=15] (251) to (248.center);
		\draw [style=edge, bend left] (242) to (252);
		\draw [style=edge, bend right] (242) to (254.center);
	\end{pgfonlayer}
\end{tikzpicture}%
    }
    \caption{Four of the rules of the ZX-calculus.
    A complete ruleset can be found in~\cite{vilmart-near-optimal-2018}.}
  \label{fig:ruleset}
\end{figure}
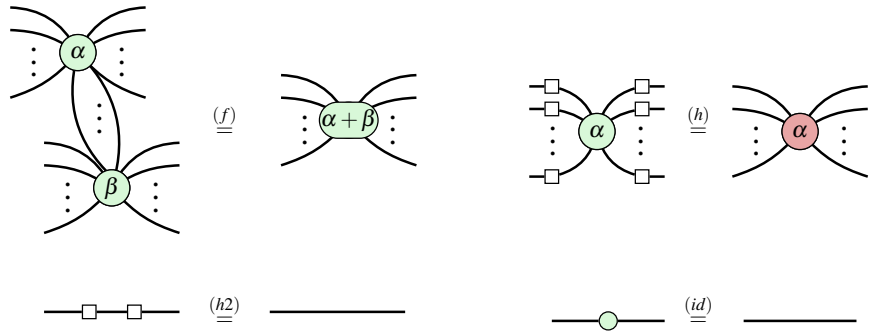

Two important properties of the ZX-calculus are \textit{universality} and \textit{completeness}.
Universality ensures that we can
represent any linear map $L: \mathbb{C}^{2^{n}}\to \mathbb{C}^{2^{m}}$ as a ZX-diagram,
while completeness means that we can always
find a series of transformations asserting that two diagrams $D_1 = D_2$ whenever their underlying
linear maps are also equal.

The ZX-calculus has been used for a variety of purposes, such as reasoning about quantum error
correction~\cite{kissinger-phase-free-2022}, and
quantum circuit optimization~\cite{kissinger-reducing-2020, backens-there-2021},
synthesis~\cite{cowtan-phase-2020, staudacher-multi-controlled-2024},
and simulation~\cite{kissinger-classical-2022}.
For those applications, usually one starts with a quantum circuit that
needs to be converted into a ZX-diagram.
Transforming a circuit into an equivalent diagram is a straightforward process, done by
going through each quantum gate and substituting it with its equivalent ZX-diagram,
as in e.g. (\ref{eq:gateExamples}).

\subsection{Diagram Reduction}\label{sec:reduction}
Once we have a ZX-diagram representation of a circuit, we can perform different rewriting
strategies on it.
Here we are interested in following a similar process as the one introduced in
\cite{duncan-graph-theoretic-2020,backens-there-2021}.
The process consists of $(1)$ transforming a ZX-diagram into a \textit{graph-like} ZX-diagram
(see~\Cref{def:graph-like} in~\Cref{app:graph})
$(2)$ carrying out a series of rewrite rules that reduce the diagram by decreasing the number of
spiders, and $(3)$ finalizing by extracting a quantum circuit from the ZX-diagram.
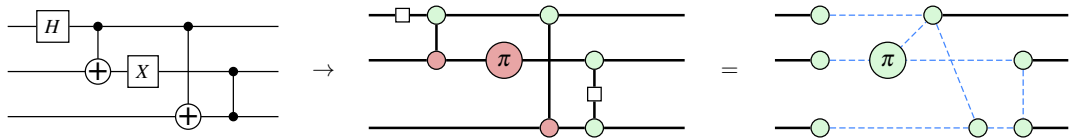
\begin{figure}
  \centering
  \scalebox{1.2}{%
    \begin{tikzpicture}
	\begin{pgfonlayer}{nodelayer}
		\node [style=Empty green] (3) at (-5, 1) {};
		\node [style=Empty red] (4) at (-5, 0) {};
		\node [style=Red Node] (5) at (-3.5, 0) {$\pi$};
		\node [style=none] (8) at (0.5, 1) {};
		\node [style=none] (9) at (0.5, 0) {};
		\node [style=borders] (17) at (-12.5, -0.25) {+};
		\node [style=none] (20) at (-8.5, 0.75) {};
		\node [style=none] (21) at (-8.5, -0.25) {};
		\node [style=none] (23) at (-14.5, 0.75) {};
		\node [style=none] (24) at (-14.5, -0.25) {};
		\node [style=txtnobold] (27) at (-7.5, -0.25) {$\rightarrow$};
		\node [style=Gate control] (28) at (-12.5, 0.75) {};
		\node [style=Quantum gate] (29) at (-13.5, 0.75) {$H$};
		\node [style=Quantum gate] (30) at (-11.5, -0.25) {$X$};
		\node [style=borders] (31) at (-10.5, -1.25) {+};
		\node [style=none] (32) at (-8.5, -1.25) {};
		\node [style=none] (33) at (-14.5, -1.25) {};
		\node [style=Gate control] (36) at (-10.5, 0.75) {};
		\node [style=Gate control] (37) at (-9.5, -0.25) {};
		\node [style=Gate control] (38) at (-9.5, -1.25) {};
		\node [style=Empty green] (39) at (-2.5, 1) {};
		\node [style=Empty green] (40) at (-1.5, 0) {};
		\node [style=Empty green] (41) at (-1.5, -1.5) {};
		\node [style=Empty red] (42) at (-2.5, -1.5) {};
		\node [style=h] (43) at (-1.5, -0.75) {};
		\node [style=none] (44) at (-6.5, 1) {};
		\node [style=none] (45) at (-6.5, 0) {};
		\node [style=none] (46) at (-6.5, -1.5) {};
		\node [style=none] (47) at (0.5, -1.5) {};
		\node [style=h] (48) at (-5.75, 1) {};
		\node [style=Empty green] (49) at (6, 1) {};
		\node [style=none] (52) at (9, 1) {};
		\node [style=none] (53) at (9, 0) {};
		\node [style=Empty green] (55) at (8, 0) {};
		\node [style=Empty green] (56) at (8, -1.5) {};
		\node [style=Empty green] (57) at (7, -1.5) {};
		\node [style=none] (59) at (2.5, 1) {};
		\node [style=none] (60) at (2.5, 0) {};
		\node [style=none] (61) at (2.5, -1.5) {};
		\node [style=none] (62) at (9, -1.5) {};
		\node [style=Empty green] (64) at (3.5, 1) {};
		\node [style=Empty green] (65) at (3.5, 0) {};
		\node [style=Green Node] (66) at (5, 0) {$\pi$};
		\node [style=Empty green] (67) at (3.5, -1.5) {};
		\node [style=txtnobold] (68) at (1.5, -0.25) {$=$};
	\end{pgfonlayer}
	\begin{pgfonlayer}{edgelayer}
		\draw (28) to (17);
		\draw (29) to (28);
		\draw (28) to (20.center);
		\draw (21.center) to (30);
		\draw (30) to (17);
		\draw (23.center) to (29);
		\draw (24.center) to (17);
		\draw (33.center) to (31);
		\draw (31) to (32.center);
		\draw (36) to (31);
		\draw (37) to (38);
		\draw [style=edge] (3) to (4);
		\draw [style=edge] (4) to (5);
		\draw [style=edge] (42) to (39);
		\draw [style=edge] (39) to (8.center);
		\draw [style=edge] (40) to (9.center);
		\draw [style=edge] (40) to (43);
		\draw [style=edge] (43) to (41);
		\draw [style=edge] (3) to (39);
		\draw [style=edge] (45.center) to (4);
		\draw [style=edge] (41) to (47.center);
		\draw [style=edge] (46.center) to (42);
		\draw [style=edge] (42) to (41);
		\draw [style=edge] (5) to (40);
		\draw [style=edge] (44.center) to (48);
		\draw [style=edge] (48) to (3);
		\draw [style=hadamard edge] (64) to (49);
		\draw [style=hadamard edge] (65) to (66);
		\draw [style=hadamard edge] (66) to (49);
		\draw [style=edge] (59.center) to (64);
		\draw [style=edge] (60.center) to (65);
		\draw [style=edge] (55) to (53.center);
		\draw [style=hadamard edge] (55) to (56);
		\draw [style=edge] (56) to (62.center);
		\draw [style=hadamard edge] (57) to (56);
		\draw [style=hadamard edge] (67) to (57);
		\draw [style=edge] (61.center) to (67);
		\draw [style=hadamard edge] (49) to (57);
		\draw [style=edge] (49) to (52.center);
		\draw [style=hadamard edge] (66) to (55);
	\end{pgfonlayer}
\end{tikzpicture}%
    }
    \caption{From quantum circuit to graph-like ZX-diagram.
      We turned the ZX-diagram into graph-like form by applying the
      rules of~\Cref{fig:ruleset} and interpreting Hadamard gates as Hadamard edges.}
  \label{fig:circToGraph}
\end{figure}

By taking ZX-diagrams into graph-like form
(which is always possible, see~\Cref{fig:circToGraph} for an example),
we can think about them as measurement patterns~\cite{backens-there-2021}
and \textit{labelled open graphs} (see \Cref{def:labelled} in~\Cref{app:graph}).
This allows us to leverage a
graph-theoretic property called generalised flow (or \textit{gflow}, see
\Cref{def:gflow} in~\Cref{app:graph})
to ensure we can extract
a circuit from a ZX-diagram efficiently and deterministically, even after simplifying
said diagram~\cite{duncan-graph-theoretic-2020}.
Starting with a quantum circuit and transforming it into a graph-like ZX-diagram ensures the
existence of gflow in our diagram, as quantum circuits themselves also have a (stronger)
notion of flow in them~\cite{duncan-graph-theoretic-2020}.

Previous works have come up with rewrite rules such as
\textit{local complementation} and \textit{pivoting}
\cite{duncan-graph-theoretic-2020},
that preserve both the graph-like structure of a diagram and the existence of gflow.
By applying these rules on a diagram repeatedly until they can no longer be applied
\footnote{Given that each rule always removes one or more vertices from the diagram if the
vertex/vertices satisfy certain conditions, we can ensure that such a diagram simplification
algorithm always terminates.},
we will arrive to a reduced diagram that still has gflow.
Once the ZX-diagram has been simplified, one performs circuit extraction on the diagram
to get a quantum circuit from it.

\subsection{Circuit Extraction}\label{sec:extraction}

Circuit extraction involves constructing a quantum circuit from a ZX-diagram by
iteratively extracting quantum gates from the output vertices of the ZX-diagram,
that we will refer to as the \emph{frontier} vertices.
Each frontier vertex corresponds to one qubit wire in the final quantum circuit, where
we place the quantum gates that we extract from said vertex.
When we extract a gate from the frontier vertices, the frontier gets simplified
and we are allowed to progressively advance it until reaching the input vertices, 
in which case we have constructed the full circuit and thus finished the extraction.
At any given point of the circuit extraction, one or more of the following gates can
be extracted into the circuit by modifying the ZX-diagram.
\begin{enumerate}
\item A phase $\alpha$ on a frontier vertex can be removed by turning it into an $\rz(\alpha)$ gate
  on the corresponding qubit.
\item One-to-one frontier vertices\begin{tikzpicture}
	\begin{pgfonlayer}{nodelayer}
		\node [style=Empty green] (0) at (0, 0) {};
		\node [style=none] (1) at (-0.75, 0) {};
		\node [style=none] (2) at (0.75, 0) {};
	\end{pgfonlayer}
	\begin{pgfonlayer}{edgelayer}
		\draw [style=edge] (0) to (2.center);
		\draw [style=hadamard edge] (1.center) to (0);
	\end{pgfonlayer}
\end{tikzpicture}

  can be removed using the $(id)$ rule and extracting a Hadamard gate due to the Hadamard edge
  on their left.
  The frontier is advanced by adding the adjacent vertex to the set of frontier vertices.
\item A Hadamard edge between frontier vertices can be removed by extracting a CZ gate.
\item The connectivity between frontier vertices and their neighbours can be simplified by
  extracting CNOT gates.
  If we form the adjacency matrix $M$ between the frontier vertices and their
  neighbours (on the left), extracting a CNOT gate with control qubit $i$ and
  target qubit $j$ equates to performing the row operation $r_i \leftarrow r_i\oplus r_j$,
  for $r_i,r_j$ rows in $M$ and $\oplus$ addition modulo $2$, thus changing the connectivity on the
  frontier vertices~\cite{duncan-graph-theoretic-2020}.
  Performing Gaussian elimination on $M$ would yield a sequence of CNOTs such
  that their associated row operations leave some rows with a single $1$ in them.
  Each frontier vertex corresponding to those reduced rows becomes a
  one-to-one \begin{tikzpicture}
	\begin{pgfonlayer}{nodelayer}
		\node [style=Empty green] (0) at (0, 0) {};
		\node [style=none] (1) at (-0.75, 0) {};
		\node [style=none] (2) at (0.75, 0) {};
	\end{pgfonlayer}
	\begin{pgfonlayer}{edgelayer}
		\draw [style=edge] (0) to (2.center);
		\draw [style=hadamard edge] (1.center) to (0);
	\end{pgfonlayer}
\end{tikzpicture}
 spider, which in
  turn allows us to advance the frontier using Case $(2)$.
\end{enumerate}

An example of each operation can be found in \Cref{fig:extractionExample}.
To fully extract a quantum circuit from a ZX-diagram, we iteratively extract gates from it following
the cases above until consuming the entire diagram.
A more detailed description of the standard circuit extraction algorithm can be found
in~\cite[Section 5]{backens-there-2021}.
It is also important to point out that gflow is preserved even after applying
the rewrites caused on the diagram by the extraction of each of the gates explained above
\cite{duncan-graph-theoretic-2020}.
In our case, we plan to use those same rewrites but with additional considerations, thus
preserving gflow too.
In particular, in~\Cref{sec:algo} we introduce a new way of generating the CNOT gates of
Case (4) above during the Gaussian elimination in such a way that they can be compiled together
into few GMS gates.

\begin{figure}
  \centering
  \scalebox{1}{%
    \begin{tikzpicture}
	\begin{pgfonlayer}{nodelayer}
		\node [style=Empty green] (0) at (-1.75, 4.75) {};
		\node [style=Empty green] (1) at (-1.75, 3.5) {};
		\node [style=Empty green] (2) at (-1.75, 1.75) {};
		\node [style=Empty green] (3) at (-1.75, 0.5) {};
		\node [style=none] (4) at (0, 4.75) {};
		\node [style=none] (5) at (0, 3.5) {};
		\node [style=none] (6) at (0, 1.75) {};
		\node [style=none] (7) at (0, 0.5) {};
		\node [style=Empty green] (9) at (-3.25, 3.5) {};
		\node [style=Empty green] (10) at (-3.25, 2.25) {};
		\node [style=Empty green] (11) at (-3.25, 0.5) {};
		\node [style=Green Node] (12) at (-3.25, 4.75) {$\alpha$};
		\node [style=none] (13) at (-4.75, 5.25) {};
		\node [style=none] (14) at (-4.75, 4.25) {};
		\node [style=none] (15) at (-4.75, 4) {};
		\node [style=none] (16) at (-4.75, 3) {};
		\node [style=none] (17) at (-4.75, 1) {};
		\node [style=none] (18) at (-4.75, 0) {};
		\node [style=none] (20) at (-4.75, 1.75) {};
		\node [style=none] (21) at (-4.75, 2.75) {};
		\node [style=txtnobold] (22) at (-4.75, 5) {$\vdots$};
		\node [style=txtnobold] (23) at (-4.75, 3.75) {$\vdots$};
		\node [style=txtnobold] (24) at (-4.75, 2.5) {$\vdots$};
		\node [style=txtnobold] (25) at (-4.75, 0.75) {$\vdots$};
		\node [style=none] (26) at (-2.25, 5.5) {};
		\node [style=none] (27) at (-1.25, 5.5) {};
		\node [style=none] (28) at (-2.25, 0) {};
		\node [style=none] (29) at (-1.25, 0) {};
		\node [style=txtnobold] (30) at (-1.75, 6) {frontier};
		\node [style=txtnobold] (31) at (0.75, 2.75) {$=$};
		\node [style=Empty green] (32) at (5.25, 4.75) {};
		\node [style=Empty green] (33) at (5.25, 3.5) {};
		\node [style=Empty green] (34) at (5.25, 1.75) {};
		\node [style=Empty green] (35) at (5.25, 0.5) {};
		\node [style=none] (36) at (8.5, 4.75) {};
		\node [style=none] (37) at (8.5, 3.5) {};
		\node [style=none] (38) at (8.5, 1.75) {};
		\node [style=none] (39) at (8.5, 0.5) {};
		\node [style=Empty green] (40) at (3.75, 3.5) {};
		\node [style=Empty green] (41) at (3.75, 2.25) {};
		\node [style=Empty green] (42) at (3.75, 0.5) {};
		\node [style=Green Node] (43) at (3.75, 4.75) {$\alpha$};
		\node [style=none] (44) at (2.25, 5.25) {};
		\node [style=none] (45) at (2.25, 4.25) {};
		\node [style=none] (46) at (2.25, 4) {};
		\node [style=none] (47) at (2.25, 3) {};
		\node [style=none] (48) at (2.25, 1) {};
		\node [style=none] (49) at (2.25, 0) {};
		\node [style=none] (51) at (2.25, 1.75) {};
		\node [style=none] (52) at (2.25, 2.75) {};
		\node [style=txtnobold] (53) at (2.25, 5) {$\vdots$};
		\node [style=txtnobold] (54) at (2.25, 3.75) {$\vdots$};
		\node [style=txtnobold] (55) at (2.25, 2.5) {$\vdots$};
		\node [style=txtnobold] (56) at (2.25, 0.75) {$\vdots$};
		\node [style=none] (57) at (4.75, 5.5) {};
		\node [style=none] (58) at (5.75, 5.5) {};
		\node [style=none] (59) at (4.75, 0) {};
		\node [style=none] (60) at (5.75, 0) {};
		\node [style=txtnobold] (61) at (5.25, 6) {frontier};
		\node [style=Gate control] (62) at (6.5, 4.75) {};
		\node [style=Gate control] (63) at (7.5, 4.75) {};
		\node [style=borders] (64) at (7.5, 1.75) {+};
		\node [style=borders] (65) at (6.5, 0.5) {+};
		\node [style=txtnobold] (66) at (9.25, 2.75) {$=$};
		\node [style=Empty green] (68) at (13.75, 3.5) {};
		\node [style=Empty green] (69) at (13.75, 1.75) {};
		\node [style=Empty green] (70) at (13.75, 0.5) {};
		\node [style=none] (71) at (21.5, 4.75) {};
		\node [style=none] (72) at (21.5, 3.5) {};
		\node [style=none] (73) at (21.5, 1.75) {};
		\node [style=none] (74) at (21.5, 0.5) {};
		\node [style=Empty green] (75) at (12.25, 3.5) {};
		\node [style=Empty green] (76) at (12.25, 2.25) {};
		\node [style=Empty green] (77) at (12.25, 0.5) {};
		\node [style=Empty green] (78) at (13.75, 4.75) {};
		\node [style=none] (79) at (10.75, 5.25) {};
		\node [style=none] (80) at (10.75, 4.25) {};
		\node [style=none] (81) at (10.75, 4) {};
		\node [style=none] (82) at (10.75, 3) {};
		\node [style=none] (83) at (10.75, 1) {};
		\node [style=none] (84) at (10.75, 0) {};
		\node [style=none] (86) at (10.75, 1.75) {};
		\node [style=none] (87) at (10.75, 2.75) {};
		\node [style=txtnobold] (88) at (10.75, 5) {$\vdots$};
		\node [style=txtnobold] (89) at (10.75, 3.75) {$\vdots$};
		\node [style=txtnobold] (90) at (10.75, 2.5) {$\vdots$};
		\node [style=txtnobold] (91) at (10.75, 0.75) {$\vdots$};
		\node [style=none] (92) at (13.25, 5.5) {};
		\node [style=none] (93) at (14.25, 5.5) {};
		\node [style=none] (94) at (13.25, 0) {};
		\node [style=none] (95) at (14.25, 0) {};
		\node [style=txtnobold] (96) at (13.75, 6) {frontier};
		\node [style=Gate control] (97) at (19, 4.75) {};
		\node [style=Gate control] (98) at (20, 4.75) {};
		\node [style=borders] (99) at (20, 1.75) {+};
		\node [style=borders] (100) at (19, 0.5) {+};
		\node [style=Quantum gate] (101) at (18, 4.75) {H};
		\node [style=Quantum gate] (102) at (16.5, 4.75) {$R_z(\alpha)$};
		\node [style=Gate control] (103) at (15, 4.75) {};
		\node [style=Gate control] (104) at (15, 3.5) {};
		\node [style=txtnobold] (105) at (-1.75, 5.25) {$v$};
		\node [style=txtnobold] (106) at (-3.25, 5.5) {$w$};
		\node [style=txtnobold] (107) at (5.25, 5.25) {$v$};
		\node [style=txtnobold] (108) at (3.75, 5.5) {$w$};
		\node [style=txtnobold] (109) at (13.75, 5.25) {$w$};
	\end{pgfonlayer}
	\begin{pgfonlayer}{edgelayer}
		\draw (0) to (4.center);
		\draw (5.center) to (1);
		\draw (2) to (6.center);
		\draw (7.center) to (3);
		\draw [style=hadamard edge] (12) to (0);
		\draw [style=hadamard edge] (0) to (10);
		\draw [style=hadamard edge] (10) to (2);
		\draw [style=hadamard edge] (2) to (11);
		\draw [style=hadamard edge] (3) to (11);
		\draw [style=hadamard edge] (12) to (13.center);
		\draw [style=hadamard edge] (12) to (14.center);
		\draw [style=hadamard edge] (1) to (9);
		\draw [style=hadamard edge] (9) to (15.center);
		\draw [style=hadamard edge] (9) to (16.center);
		\draw [style=hadamard edge] (11) to (17.center);
		\draw [style=hadamard edge] (11) to (18.center);
		\draw [style=hadamard edge] (10) to (21.center);
		\draw [style=hadamard edge] (10) to (20.center);
		\draw [style=discon] (28.center)
			 to (26.center)
			 to (27.center)
			 to (29.center)
			 to cycle;
		\draw (37.center) to (33);
		\draw [style=hadamard edge] (43) to (32);
		\draw [style=hadamard edge] (41) to (34);
		\draw [style=hadamard edge] (34) to (42);
		\draw [style=hadamard edge] (35) to (42);
		\draw [style=hadamard edge] (43) to (44.center);
		\draw [style=hadamard edge] (43) to (45.center);
		\draw [style=hadamard edge] (33) to (40);
		\draw [style=hadamard edge] (40) to (46.center);
		\draw [style=hadamard edge] (40) to (47.center);
		\draw [style=hadamard edge] (42) to (48.center);
		\draw [style=hadamard edge] (42) to (49.center);
		\draw [style=hadamard edge] (41) to (52.center);
		\draw [style=hadamard edge] (41) to (51.center);
		\draw [style=discon] (59.center)
			 to (57.center)
			 to (58.center)
			 to (60.center)
			 to cycle;
		\draw (63) to (64);
		\draw (62) to (65);
		\draw (32) to (62);
		\draw (62) to (63);
		\draw (63) to (36.center);
		\draw (34) to (64);
		\draw (64) to (38.center);
		\draw (35) to (65);
		\draw (65) to (39.center);
		\draw (72.center) to (68);
		\draw [style=hadamard edge] (76) to (69);
		\draw [style=hadamard edge] (69) to (77);
		\draw [style=hadamard edge] (70) to (77);
		\draw [style=hadamard edge] (78) to (79.center);
		\draw [style=hadamard edge] (78) to (80.center);
		\draw [style=hadamard edge] (68) to (75);
		\draw [style=hadamard edge] (75) to (81.center);
		\draw [style=hadamard edge] (75) to (82.center);
		\draw [style=hadamard edge] (77) to (83.center);
		\draw [style=hadamard edge] (77) to (84.center);
		\draw [style=hadamard edge] (76) to (87.center);
		\draw [style=hadamard edge] (76) to (86.center);
		\draw [style=discon] (94.center) to (92.center);
		\draw [style=discon] (92.center) to (93.center);
		\draw [style=discon] (93.center) to (95.center);
		\draw [style=discon] (95.center) to (94.center);
		\draw (98) to (99);
		\draw (97) to (100);
		\draw (97) to (98);
		\draw (98) to (71.center);
		\draw (69) to (99);
		\draw (99) to (73.center);
		\draw (70) to (100);
		\draw (100) to (74.center);
		\draw (78) to (102);
		\draw (102) to (101);
		\draw (101) to (97);
		\draw [style=hadamard edge] (43) to (33);
		\draw (103) to (104);
		\draw [style=hadamard edge] (33) to (41);
		\draw [style=hadamard edge] (68) to (76);
		\draw [style=hadamard edge] (1) to (10);
		\draw [style=hadamard edge] (12) to (1);
	\end{pgfonlayer}
\end{tikzpicture}%
    }
    \caption{An example of circuit extraction.
      In the first equation we simplify the connectivity of $v$ using two CNOT gates.
      In the second equation we extract the Hadamard gate corresponding to the Hadamard
      edge $(v,w)$, an $R_Z$ due to the phase of $w$, and a $CZ$ due to the Hadamard edge between
      $w$ and the second frontier vertex.
    These steps would be repeated until there is no more diagram to be extracted.}
  \label{fig:extractionExample}
\end{figure}
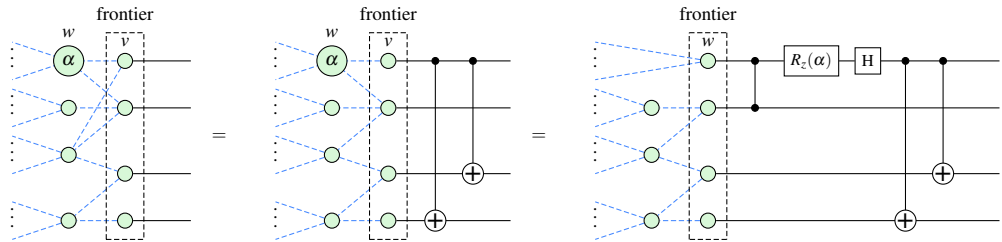

\section{Compilation for Ion Trap Hardware}~\label{sec:comp}

In this section we introduce our methods to compile quantum circuits using global gates.
Our approach consists of modifying the circuit extraction routine explained in
\Cref{sec:extraction} in two ways.
First, we add extra logic to the overall extraction algorithm by checking how
gates can be compiled efficiently into few GMS gates as they are being extracted.
This is done by performing certain local rewrites in the circuit that exploit how the single- and
two-qubit gates behave for better compilation.
Second, we introduce a new way of generating the CNOTs required to simplify the frontier in
Case (4) of the circuit extraction algorithm of~\Cref{sec:extraction}.
This new method consists of generating CNOTs in layers that can be compiled with a single GMS
gate by encoding the task of simplifying the frontier with this restriction
as a linear program.
These two methods fit together into one algorithm by performing circuit extraction
while doing local rewrites
whenever new gates are extracted, and doing the frontier simplification with our new method whenever
we have to perform Case (4).
The layer of CNOT gates generated by our new method is then also compiled using the new rewrites we
introduce in this section.

\subsection{Grouping Entangling Gates into Global Gates}\label{sec:equivalences}
We begin by mentioning the circuit equivalences we use
for compiling entangling gates as global gates.
In \Cref{fig:circuitIdentities} we showcase examples of the circuit equivalences presented in this
section in quantum circuit notation.
\begin{figure}
  \centering
  \scalebox{0.9}{%
    \input{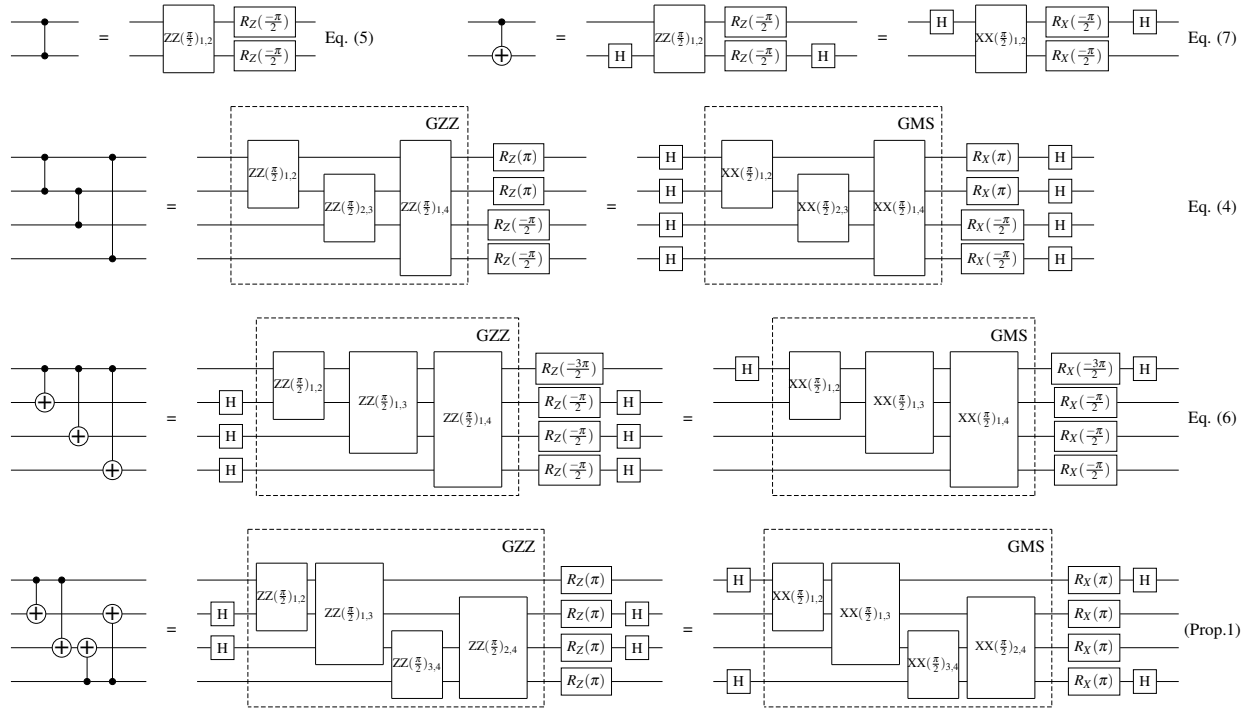}%
  }
  \caption{Examples of the circuit equivalences given in this section at a glance.
    In the first row we show how to implement a CZ and a CNOT using two-qubit rotations.
    In the subsequent rows we show an example of implementing a layer of CZ gates,
  a fanout gate, and a layer of commuting CNOTs using a single GZZ or GMS gate.}
  \label{fig:circuitIdentities}
\end{figure}

\paragraph{CZ layers}\mbox{}\\
CZ layers refer to series of CZ gates on $n$ qubits, commonly found
on circuits for graph state preparation
\cite{wetering-zx-calculus-2020}
and on Clifford circuits in normal form~\cite{duncan-graph-theoretic-2020}.
For our purposes, we are interested in knowing how to compile CZ layers efficiently since
we generate these layers when doing circuit extraction (see~\Cref{sec:extraction}).
Using a symmetric binary matrix $A$ to represent where the CZ gates are placed in the layer,
we can see how we can implement an arbitrary CZ layer with a single GZZ (or GMS)
plus single-qubit gates. Let $ c_i = \sum_{j=1}^n A_{i,j}$ we have

\begin{gather}\label{eq:cztogzz}
  \begin{aligned}
  \prod_{i>j\in A} \cz_{i,j} &= \prod_{i>j\in A} \rz(-\pi/2)_i\ \rz(-\pi/2)_j\ \zz(\pi/2)_{i,j}\\
                             &=\prod_{i=1}^{n} \rz(-c_i\pi/2)_i \prod_{k>j\in A} \zz(\pi/2)_{k,j}\\
                             &=\prod_{i=1}^{n} \rz(-c_i\pi/2)_i\ \gzz_A(\pi/2) \\
                             &= \h^{\otimes n}\ \prod_{i=1}^{n} \rx(-c_i\pi/2)_i\ \gms_A(\pi/2) \  \h^{\otimes n} ,
\end{aligned}
\end{gather}
where we have used the equality
\begin{equation}
  \label{eq:cztozz}
  \cz_{i,j} =  \rz(-\pi/2)_i\ \rz(-\pi/2)_j\ \zz(\pi/2)_{i,j}.
\end{equation}
Notice how for qubits not involved in the CZ layer the rotation gates $\rz$ and $\rx$ would have
rotation angle zero.
CZ layers have the convenient property that they can always be compiled with a single global gate.
This is not the case however for CNOT layers, where this is only possible when the CNOTs are
arranged in a particular way, as we will see now.

\paragraph{Fanout gates}\mbox{}\\
Another sequence of entangling gates that can be implemented with one global gate is the
\emph{fanout} gate.
A fanout gate consists of a sequence of CNOTs that all share the same control qubit.
Fanout gates are a common occurrence when compiling diagonal
unitaries~\cite{basler-synthesis-2023}
and syndrome measurements useful for
error correction and mitigation~\cite{Nielsen-Chuang-2010}.
We use the notation $\fo^{(i)}_S$ for a fanout gate with control qubit $i$ and target qubits
$j\in S$ for $S$ a subset of the available qubits.
We can compile this gate with a single global gate (and some single-qubit gates)
as follows, for $A_{S}^{(i)}$ the matrix with all zeroes except for the elements located
in row $i$ and with column index in $S$:
\begin{gather}\label{eq:fotogzz}
  \begin{aligned}
    \fo^{(i)}_S &= \prod_{j\in S}\ \cnot_{i,j}\\
                &= \prod_{j\in S}\ \h_j\ \cz_{i,j}\ \h_j \\
                &= \rz(-|S|\pi/2)_i\  \prod_{j\in S}\ \h_j\ \rz(- \pi/2)_j \ \zz(\pi/2)_{i,j}\  \h_j\\
                &= \rz(-|S|\pi/2)_i\  \biggl( \prod_{j\in S}\ \h_j\ \rz(- \pi/2)_j\ \biggr) \biggl( \prod_{l\in S}\zz(\pi/2)_{i,l} \biggr)\  \prod_{k\in S} \ \h_k\\
                &= \rz(-|S|\pi/2)_i\  \biggl( \prod_{j\in S}\ \h_j\ \rz(- \pi/2)_j\ \biggr) \gzz_{A_{S}^{(i)}}(\pi/2)\ \prod_{k\in S} \ \h_k\\
                &= \h_i\ \rx(-|S|\pi/2)_i\  \prod_{j\in S}\ \rx(- \pi/2)_j\  \gms_{A_{S}^{(i)}}(\pi/2)\  \h_i.
\end{aligned}
\end{gather}

\paragraph{Commuting CNOT layers}\mbox{}\\
Lastly, we show how to compile more general layers of CNOT gates with a single global gate.
Layers of CNOT gates, also referred to as CNOT circuits, are a class of quantum circuits that
generate what are called \emph{linear reversible circuits}~\cite{patel-optimal-2008, minimalcnot}.
The limiting factor when grouping CNOTs into a single global interaction comes from needing to
place Hadamard gates into the circuit when turning each CNOT into a $\zz$ or $\xx$ rotation
following the circuit equivalences
\begin{gather}
  \label{eq:cnottorotations}
  \begin{aligned}
    \cnot_{i,j} &=\h_j\ \rz(-\pi/2)_j\ \rz(-\pi/2)_i\ \zz(\pi/2)_{i,j}\ \h_j\\
                &=  \h_i\ \rx(-\pi/2)_j\ \rx(-\pi/2)_i\ \xx(\pi/2)_{i,j}\ \h_i.
  \end{aligned}
\end{gather}
Consequently, these Hadamard gates get in the way when trying to group the two-qubit rotations
on a single $\gms$ or $\gzz$.
In particular, given that we have to conjugate either the control or target qubits
(depending on the axis of rotation we choose) with Hadamards, what determines if a sequence
of CNOTs can be compiled with a single global gate is whether the Hadamards can be removed or
pushed all the way to the left and right.

Recall from~\Cref{sec:extraction} that we construct the output quantum circuits from
a ZX-diagram by iteratively extracting gates.
This means that rather than starting with a series of CNOT layers to be compiled as global gates,
we have some freedom to choose how these layers of CNOTs look like in the first place.
By ensuring that these layers have one property that
guarantees they can be compiled with a single global gate, we are able to streamline the
linear program (\lp) of
\Cref{sec:lp} and the peephole optimization algorithm of~\Cref{sec:peephole}.
This property is \emph{commutativity}.
In particular, if we are able to simplify the adjacency matrix $M$ of the frontier using only
commuting CNOTs until one or more vertices are extractable, we will be able to merge those CNOTs
into a single GMS (plus some single-qubit gates).
We first show in~\Cref{prop:commutingCnotsToGlobal}
that such a layer of commuting CNOTs can be compiled with a single global gate.

\begin{proposition}\label{prop:commutingCnotsToGlobal}
  Let $\cnot_{i_n,j_n}\dots\cnot_{i_1,j_1}$ be a linear reversible circuit consisting of
  commuting CNOT gates.
  This circuit can always be compiled with one global gate plus single-qubit gates.
\end{proposition}
\begin{proof}
  Without loss of generality (given that a GMS and GZZ are equivalent up to conjugation by
  Hadamards), we show that this is the case with the GMS, using the decomposition of the CNOT gate
  into an $\xx$ rotation and single-qubit gates from
  \Cref{eq:cnottorotations}.
  A linear reversible circuit consists of commuting CNOT gates if,
  when grouping the qubits that act as controls and as targets for the CNOTs in the
  circuit into two separate sets, these sets are disjoint.
  When this is the case, we show now how we can implement this circuit with one GMS and single-qubit
  gates.
  
  Start by decomposing all the CNOT gates individually using~\Cref{eq:cnottorotations}.
  Then, notice that $(1)$ single-qubit $\rx$ gates can be commuted past $\xx$
  gates, and $(2)$
  the Hadamards only get in the way of creating a GMS
  if the control qubit of one CNOT is also involved in another CNOT.
  If this is the case then the qubit serves as control in both CNOTs (given that these CNOTs commute),
  then we would have two Hadamards next to each other due to the decomposition
  of~\Cref{eq:cnottorotations}, in which case they
  cancel out and we are left only with the outermost Hadamards.
  This allows us to have all Hadamards only to the left and right of the circuit
  (none in between $\xx$ gates), $\rx$ gates can be commuted to the left or right of all $\xx$
  gates, and the $\xx$ interactions are together which we can then
  compile with a single GMS.
\end{proof}
An example of a commuting CNOT layer compiled with a GZZ and a GMS can be found in
\Cref{fig:circuitIdentities}.
We are also interested in representing such a layer of commuting CNOTs in a succinct manner
as a matrix to reason about its effect when extracting commuting CNOT layers
during circuit extraction.

\begin{proposition}\label{prop:commutingG}
  During circuit extraction of an $n$-qubit quantum circuit,
  given a frontier with corresponding adjacency matrix $M$,
  extracting a layer of commuting CNOTs yields a frontier with an updated adjacency
  matrix $M' = GM$, where $G$ is an $n\times n$ binary matrix representing the layer
  of commuting CNOTs.
  For $G$ to capture the effect of simplifying the frontier with a layer of commuting CNOTs,
  it must satisfy
  $(1)$ all entries in the diagonal are $1$ and $(2)$ if there is a $1$ on
  an off-diagonal entry $G_{i,j}$ then all off-diagonal entries on row $j$ must be $0$.
\end{proposition}
\begin{proof}
  Recall from~\Cref{sec:extraction} that the action of extracting a $\cnot_{i,j}$
  having a frontier with connectivity $M$ produces
  the row operation $r_i \leftarrow r_i\oplus r_j$ in $M$.
  This in turn lets us relate an \emph{elementary matrix} to extracted CNOTs.
  In particular for a $\cnot_{i,j}$ we have the elementary matrix
  $B$ which is a $n\times n$ Identity matrix with an additional $1$ in position $B_{i,j}$.
  From~\Cref{prop:commutingCnotsToGlobal} we can see that in a commuting CNOT layer, a qubit
  cannot serve both as a control and as a target of different CNOTs, otherwise these would not
  commute in general.
  Then, let $G=B_k\dots B_1$ be the matrix resulting from composing the elementary matrices of
  a sequence of commuting CNOTs, we have for every elementary matrix $B$ that if $B_{i,j}=1$
  (for $i\ne j$) then the other elementary matrices must necessarily have all zeroes in the off
  diagonal elements of row $j$, otherwise their corresponding CNOTs would not commute.
  It is easy to see that multiplying these matrices $B_k\dots B_1$
  would result in a $G$ with ones in the diagonal
  and with the property that for off diagonal elements if $G_{i,j}=1$ then all off diagonal elements
  of row $j$ remain $0$.
\end{proof}

We now leverage the above circuit identities and results to create a circuit compilation
algorithm.

\subsection{Creating Quantum Circuits with Global Gates}\label{sec:algo}

We present here our compilation algorithm that takes as input an arbitrary quantum circuit
and outputs an equivalent circuit consisting
of a gate-set supported by ion trap quantum hardware.
In particular, we are interested in (1) using the GMS gate as the entangling operation and (2)
minimizing the number of GMS gates needed to implement the circuit.
On a high level, the algorithm works by performing the following steps:
\begin{enumerate}
\item Transforming the quantum circuit into a (graph-like) ZX-diagram.
\item Simplifying the diagram with gflow-preserving rules as discussed in~\Cref{sec:reduction}
\item Extracting a compiled circuit out of the ZX-diagram by performing a modification of
  the base circuit extraction algorithm from~\Cref{sec:extraction}.
We change the base circuit extraction algorithm in two ways.
First, we change the way in which the connectivity of the frontier is simplified by encoding
the problem as a LP that ensures all CNOTs that are extracted in each round fit in
a single GMS.
Second, we perform peephole optimization rewrites during the circuit extraction to group the
CNOTs and CZs into GMS gates, and use extracted Hadamard and rotation gates to make circuit
rewrites that ensure a lower total gate count.
The peephole rewrites are performed as we extract CNOT, CZ, $\rz$, and Hadamard gates
from the diagram as it is done in~\Cref{sec:extraction}, with the exception of CNOT gates,
which are extracted following the output of our LP.
\end{enumerate}
We refer to the algorithm as \texttt{gms_compiler}. Pseudocode for it can be found in~\Cref{app:pseudocode}.
We now explain in detail how the linear program and the peephole optimization works.

\subsubsection{Linear Program}\label{sec:lp}

When simplifying a frontier during circuit extraction (see Case (4) in~\Cref{sec:extraction}),
we have a set of frontier vertices and their neighbourhood.
Our objective is to turn as many frontier vertices
as possible into one-to-one vertices by extracting a sequence of CNOTs
that can be compiled with a single GMS.
We achieve this by encoding the problem as a linear program.
Here, we present its formulation on a high level, as certain constraints,
such as conditional assignment and variable multiplication,
require additional (less intuitive) constraints and variables in order to be linearized properly.
Given an input $n\times m$ adjacency matrix $M$ representing the connections between the frontier
vertices and their neighborhood, the LP is as follows,
for $c\in \mathbb{N},z_i,x_{ij},G_{ij}\in \{0,1\}$ for all $ i,k \in \{1,\dots,n\}, j \in \{1,\dots,m\}$.

\begin{gather}
  \label{eq:lp}
  \begin{aligned}
    \text{maximize} & \quad \sum_{i=1}^n n z_i \ \ - c&  \\
    \text{subject to} & \quad \sum_{i=1}^n z_i  &&\geq  1 & \quad \text{(c1)}\\
                    & \quad z_i &&= \begin{cases}
                                      1, & \text{if } \sum_{j=1}^m x_{ij} = 1\\
                                      0, & \text{otherwise}
                                    \end{cases} & \quad \text{(c2)}\\
                    & \quad x_{ij} &&= [GM]_{ij} & \quad \text{(c3)}\\
                    & \quad G_{i,i} &&= 1 & \quad \text{(c4)}\\
                    & \quad (1-G_{kj}) + \prod_{i \ne j} (1- G_{ji}) &&\geq 1 & \quad \text{(c5)}\\
                    & \quad c &&= \sum_{i\ne j} G_{i,j} & \quad \text{(c6)}
  \end{aligned}
\end{gather}

Constraints (c1) and (c2) express how we count the number of extractable vertices in the frontier
by defining a binary variable $z_i$ with $z_i = 1$ if row $i$ has a single $1$
(thus a single neighbour) in it, and we enforce that at least there is one extractable vertex.
Constraint (c3) defines the variables $x_{ij}$, which correspond to the entries of the matrix
$X=GM$ that represents the state of the frontier $M$ after applying the matrix $G$ which encodes
the row operations.
Constraints (c4) and (c5) model how we encode the possible row operations that can be done with one
global gate encoded in $G$.
These are exactly the same conditions of \Cref{prop:commutingG}.
Constraint (c6) encodes in $c$ the number of CNOTs that the program is using to reduce $M$.
Finally, the objective function aims to maximize the number of reduced rows, with a penalty on the
CNOT count used for that.
We add a weight of $n$ to each reduced row, since as worst-case one would need $n-1$ CNOTs to reduce
one row.
An equivalent \lp with all constraints linearized can be found in~\Cref{app:lp}.
In there, we have linearized the conditional assignment in (c2) by introducing additional variables
and constraints.
Similarly, doing the (binary) matrix multiplication $X=GM$ of (c3) requires performing XORs, which also need to
be linearized.
We also linearize the variable multiplication of (c5).
As a proof of correctness, we show that one such reduction of a frontier with only commuting CNOTs
always exists.

\begin{proposition}\label{prop:commutingExtract}
  During circuit extraction, given that the ZX-diagram has a gflow,
  we can always simplify one vertex in the frontier to make it a one-to-one vertex
  using a series of commuting CNOTs.
  In particular, a single fanout gate is sufficient to simplify at least one vertex from the frontier.
\end{proposition}
\begin{proof}
  Let $(G,I,O,\lambda)$ be the labelled open graph corresponding to the ZX-diagram
  we are extracting from, and $F$ the set of frontier vertices.
  We start by picking a vertex $v\in V$ that is maximal in the gflow order.
  This gives us a vertex with
  the property $g(v) \subseteq F$~\cite[Lemma 5.4]{backens-there-2021}.
  From here, constructing a \emph{maximally delayed gflow}
  $(g,\prec)$ makes it so the correction set of $v$ has the property
  Odd$(g(v)) = \{v\}$
  \cite[Prop. 3.14]{backens-there-2021}.
  Having Odd$(g(v)) = \{v\}$ means that we have a subset of the frontier vertices ($g(v)$)
  that we can work with in order to create a one-to-one vertex in the frontier.
  To do so, we start by
  constructing the adjacency matrix of $g(v)$ with their neighbours.
  We pick any row (call it $r_w$) associated to a frontier vertex $w\in g(v)$, and add
  all other rows (modulo $2$) to it.
  This makes $r_w$ to have all zeroes except for a $1$ in the
  column corresponding to $v$.
  This in turn means that the frontier vertex corresponding to $r_w$ would be a one-to-one
  spider with $v$ as its unique neighbour and thus we can advance the frontier.
  Given that $r_w$ represents a frontier vertex, it has an associated qubit with some index $i$.
  We have to notice that adding all the rows in the adjacency matrix to $r_w$
  is achieved by extracting the fanout gate $\fo^{(i)}_{S}$ for $S$ the set of indices of the qubits
  associated with the vertices in $g(v)\setminus \{w\}$.
  All the CNOTS in a fanout commute and thus it can be compiled with a single global gate.
\end{proof}
\subsubsection{Peephole Optimization}\label{sec:peephole}
We build on top of the standard circuit extraction algorithm presented
in~\Cref{sec:extraction} by adding additional logic with the objective of merging extracted
two-qubit gates into GMS gates and keeping the amount of GMS gates as low as possible.
We also take advantage of extracted single-qubit gates to better compile the GMS gates and keep
the total gate count low.
During peephole optimization we do not assume that the layers of CNOTs that are extracted
are necessarily a commuting layer, so the methods explained below also work with a different
strategy for the CNOT extraction part than the LP method we introduced.
For example, PyZX~\cite{pyzx} uses the CNOT circuit synthesis algorithm of Patel et al.
\cite{patel-optimal-2008} (based on matrix decomposition)
to extract CNOT gates in layers that are not necessarily commuting.
The general idea is to progressively build a GMS gate with the incoming gates until we reach a gate
that cannot fit in it, in which case we stop constructing the current GMS and begin building a new
one starting with said gate, to which incoming gates will be added if possible and so on.
We repeat this process until the extraction is finished.
We now explain our methods, distinguishing between each gate that can be extracted:
\begin{enumerate}
\item Hadamard gates that are extractable due to naturally appearing one-to-one frontier vertices
  (which did not require prior simplification via CNOTs) and due to performing pivoting on the
  frontier are treated the same as extracted $\rz$ gates.
  If the qubit they appear on is not involved in the GMS that is currently being built,
  we place these
  gates in the circuit to the right of that GMS, ensuring they do not get in the way of upcoming
  gates.
  If the qubit participates in the GMS, we place them in front of the GMS,
  given that they cannot commute past a GMS.
  Here is an example of the former:
  \begin{equation*}
    \scalebox{1}
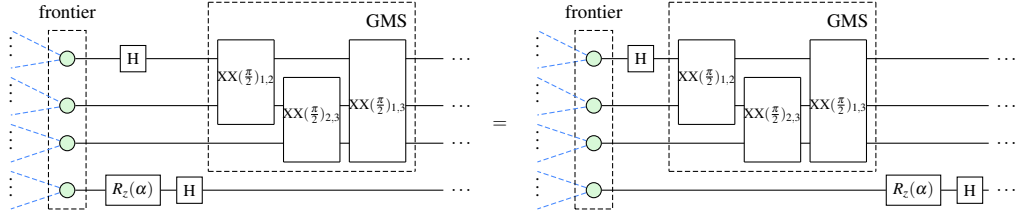
{\begin{tikzpicture}
	\begin{pgfonlayer}{nodelayer}
		\node [style=Empty green] (1) at (-11.5, 0.5) {};
		\node [style=Empty green] (2) at (-11.5, -0.5) {};
		\node [style=Empty green] (10) at (-11.5, -1.75) {};
		\node [style=Empty green] (11) at (-11.5, 1.75) {};
		\node [style=none] (12) at (-13, 2.5) {};
		\node [style=none] (13) at (-13, 1.5) {};
		\node [style=none] (14) at (-13, 1.25) {};
		\node [style=none] (15) at (-13, 0.25) {};
		\node [style=none] (16) at (-13, -1.25) {};
		\node [style=none] (17) at (-13, -2.25) {};
		\node [style=none] (18) at (-13, -1) {};
		\node [style=none] (19) at (-13, 0) {};
		\node [style=txtnobold] (20) at (-13, 2.25) {$\vdots$};
		\node [style=txtnobold] (21) at (-13, 1) {$\vdots$};
		\node [style=txtnobold] (22) at (-13, -0.25) {$\vdots$};
		\node [style=txtnobold] (23) at (-13, -1.5) {$\vdots$};
		\node [style=none] (24) at (-12, 2.5) {};
		\node [style=none] (25) at (-11, 2.5) {};
		\node [style=none] (26) at (-12, -2.25) {};
		\node [style=none] (27) at (-11, -2.25) {};
		\node [style=txtnobold] (28) at (-11.5, 3) {frontier};
		\node [style=none] (29) at (-7.5, 2.25) {};
		\node [style=none] (30) at (-6, 2.25) {};
		\node [style=none] (31) at (-7.5, 0) {};
		\node [style=none] (32) at (-6, 0) {};
		\node [style=smallest txt] (33) at (-6.75, 1.25) {XX$(\frac{\pi}{2})_{1,2}$};
		\node [style=none] (34) at (-7.5, 0.5) {};
		\node [style=none] (35) at (-6, 0.5) {};
		\node [style=none] (36) at (-7.5, 1.75) {};
		\node [style=none] (37) at (-6, 1.75) {};
		\node [style=none] (38) at (-5.75, 1.25) {};
		\node [style=none] (39) at (-4.25, 1.25) {};
		\node [style=none] (40) at (-5.75, -1) {};
		\node [style=none] (41) at (-4.25, -1) {};
		\node [style=smallest txt] (42) at (-5, 0.25) {XX$(\frac{\pi}{2})_{2,3}$};
		\node [style=none] (43) at (-5.75, -0.5) {};
		\node [style=none] (44) at (-4.25, -0.5) {};
		\node [style=none] (45) at (-5.75, 0.5) {};
		\node [style=none] (46) at (-4.25, 0.5) {};
		\node [style=none] (47) at (-4, 2.25) {};
		\node [style=none] (48) at (-2.5, 2.25) {};
		\node [style=none] (49) at (-4, -1) {};
		\node [style=none] (50) at (-2.5, -1) {};
		\node [style=smallest txt] (51) at (-3.25, 0.5) {XX$(\frac{\pi}{2})_{1,3}$};
		\node [style=none] (52) at (-4, -0.5) {};
		\node [style=none] (53) at (-2.5, -0.5) {};
		\node [style=none] (54) at (-4, 1.75) {};
		\node [style=none] (55) at (-2.5, 1.75) {};
		\node [style=none] (56) at (-4, 0.5) {};
		\node [style=none] (57) at (-2.5, 0.5) {};
		\node [style=none] (60) at (-7.75, 3.25) {};
		\node [style=none] (61) at (-2.25, -1.25) {};
		\node [style=none] (62) at (-2.25, 3.25) {};
		\node [style=none] (63) at (-7.75, -1.25) {};
		\node [style=txtnobold] (64) at (-3, 2.75) {GMS};
		\node [style=Quantum gate] (65) at (-8.25, -1.75) {H};
		\node [style=Quantum gate] (66) at (-9.75, -1.75) {$R_z(\alpha)$};
		\node [style=none] (67) at (-1.5, 1.75) {};
		\node [style=none] (68) at (-1.5, 0.5) {};
		\node [style=none] (69) at (-1.5, -0.5) {};
		\node [style=none] (71) at (-1.5, -1.75) {};
		\node [style=Quantum gate] (72) at (-9.75, 1.75) {H};
		\node [style=txtnobold] (73) at (-1, 1.75) {$\dots$};
		\node [style=txtnobold] (74) at (-1, 0.5) {$\dots$};
		\node [style=txtnobold] (75) at (-1, -0.5) {$\dots$};
		\node [style=txtnobold] (76) at (-1, -1.75) {$\dots$};
		\node [style=txtnobold] (77) at (0, 0) {$=$};
		\node [style=Empty green] (78) at (2.5, 0.5) {};
		\node [style=Empty green] (79) at (2.5, -0.5) {};
		\node [style=Empty green] (80) at (2.5, -1.75) {};
		\node [style=Empty green] (81) at (2.5, 1.75) {};
		\node [style=none] (82) at (1, 2.5) {};
		\node [style=none] (83) at (1, 1.5) {};
		\node [style=none] (84) at (1, 1.25) {};
		\node [style=none] (85) at (1, 0.25) {};
		\node [style=none] (86) at (1, -1.25) {};
		\node [style=none] (87) at (1, -2.25) {};
		\node [style=none] (88) at (1, -1) {};
		\node [style=none] (89) at (1, 0) {};
		\node [style=txtnobold] (90) at (1, 2.25) {$\vdots$};
		\node [style=txtnobold] (91) at (1, 1) {$\vdots$};
		\node [style=txtnobold] (92) at (1, -0.25) {$\vdots$};
		\node [style=txtnobold] (93) at (1, -1.5) {$\vdots$};
		\node [style=none] (94) at (2, 2.5) {};
		\node [style=none] (95) at (3, 2.5) {};
		\node [style=none] (96) at (2, -2.25) {};
		\node [style=none] (97) at (3, -2.25) {};
		\node [style=txtnobold] (98) at (2.5, 3) {frontier};
		\node [style=none] (99) at (4.75, 2.25) {};
		\node [style=none] (100) at (6.25, 2.25) {};
		\node [style=none] (101) at (4.75, 0) {};
		\node [style=none] (102) at (6.25, 0) {};
		\node [style=smallest txt] (103) at (5.5, 1.25) {XX$(\frac{\pi}{2})_{1,2}$};
		\node [style=none] (104) at (4.75, 0.5) {};
		\node [style=none] (105) at (6.25, 0.5) {};
		\node [style=none] (106) at (4.75, 1.75) {};
		\node [style=none] (107) at (6.25, 1.75) {};
		\node [style=none] (108) at (6.5, 1.25) {};
		\node [style=none] (109) at (8, 1.25) {};
		\node [style=none] (110) at (6.5, -1) {};
		\node [style=none] (111) at (8, -1) {};
		\node [style=smallest txt] (112) at (7.25, 0.25) {XX$(\frac{\pi}{2})_{2,3}$};
		\node [style=none] (113) at (6.5, -0.5) {};
		\node [style=none] (114) at (8, -0.5) {};
		\node [style=none] (115) at (6.5, 0.5) {};
		\node [style=none] (116) at (8, 0.5) {};
		\node [style=none] (117) at (8.25, 2.25) {};
		\node [style=none] (118) at (9.75, 2.25) {};
		\node [style=none] (119) at (8.25, -1) {};
		\node [style=none] (120) at (9.75, -1) {};
		\node [style=smallest txt] (121) at (9, 0.5) {XX$(\frac{\pi}{2})_{1,3}$};
		\node [style=none] (122) at (8.25, -0.5) {};
		\node [style=none] (123) at (9.75, -0.5) {};
		\node [style=none] (124) at (8.25, 1.75) {};
		\node [style=none] (125) at (9.75, 1.75) {};
		\node [style=none] (126) at (8.25, 0.5) {};
		\node [style=none] (127) at (9.75, 0.5) {};
		\node [style=none] (128) at (4.5, 3.25) {};
		\node [style=none] (129) at (10, -1.25) {};
		\node [style=none] (130) at (10, 3.25) {};
		\node [style=none] (131) at (4.5, -1.25) {};
		\node [style=txtnobold] (132) at (9.25, 2.75) {GMS};
		\node [style=Quantum gate] (133) at (12.5, -1.75) {H};
		\node [style=Quantum gate] (134) at (11, -1.75) {$R_z(\alpha)$};
		\node [style=none] (135) at (13, 1.75) {};
		\node [style=none] (136) at (13, 0.5) {};
		\node [style=none] (137) at (13, -0.5) {};
		\node [style=none] (138) at (13, -1.75) {};
		\node [style=Quantum gate] (139) at (3.75, 1.75) {H};
		\node [style=txtnobold] (140) at (13.5, 1.75) {$\dots$};
		\node [style=txtnobold] (141) at (13.5, 0.5) {$\dots$};
		\node [style=txtnobold] (142) at (13.5, -0.5) {$\dots$};
		\node [style=txtnobold] (143) at (13.5, -1.75) {$\dots$};
	\end{pgfonlayer}
	\begin{pgfonlayer}{edgelayer}
		\draw [style=hadamard edge] (11) to (12.center);
		\draw [style=hadamard edge] (11) to (13.center);
		\draw [style=hadamard edge] (10) to (16.center);
		\draw [style=hadamard edge] (10) to (17.center);
		\draw [style=discon] (26.center)
			 to (24.center)
			 to (25.center)
			 to (27.center)
			 to cycle;
		\draw [style=hadamard edge] (14.center) to (1);
		\draw [style=hadamard edge] (1) to (15.center);
		\draw [style=hadamard edge] (19.center) to (2);
		\draw [style=hadamard edge] (2) to (18.center);
		\draw (29.center) to (30.center);
		\draw (29.center) to (31.center);
		\draw (31.center) to (32.center);
		\draw (32.center) to (30.center);
		\draw (38.center) to (39.center);
		\draw (38.center) to (40.center);
		\draw (40.center) to (41.center);
		\draw (41.center) to (39.center);
		\draw (47.center) to (48.center);
		\draw (47.center) to (49.center);
		\draw (49.center) to (50.center);
		\draw (50.center) to (48.center);
		\draw [style=discon] (60.center) to (62.center);
		\draw [style=discon] (62.center) to (61.center);
		\draw [style=discon] (60.center) to (63.center);
		\draw [style=discon] (63.center) to (61.center);
		\draw (37.center) to (54.center);
		\draw (35.center) to (45.center);
		\draw (46.center) to (56.center);
		\draw (66) to (65);
		\draw (10) to (66);
		\draw (65) to (71.center);
		\draw (44.center) to (52.center);
		\draw (53.center) to (69.center);
		\draw (57.center) to (68.center);
		\draw (55.center) to (67.center);
		\draw (2) to (43.center);
		\draw (1) to (34.center);
		\draw (11) to (72);
		\draw (72) to (36.center);
		\draw [style=hadamard edge] (81) to (82.center);
		\draw [style=hadamard edge] (81) to (83.center);
		\draw [style=hadamard edge] (80) to (86.center);
		\draw [style=hadamard edge] (80) to (87.center);
		\draw [style=discon] (96.center)
			 to (94.center)
			 to (95.center)
			 to (97.center)
			 to cycle;
		\draw [style=hadamard edge] (84.center) to (78);
		\draw [style=hadamard edge] (78) to (85.center);
		\draw [style=hadamard edge] (89.center) to (79);
		\draw [style=hadamard edge] (79) to (88.center);
		\draw (99.center) to (100.center);
		\draw (99.center) to (101.center);
		\draw (101.center) to (102.center);
		\draw (102.center) to (100.center);
		\draw (108.center) to (109.center);
		\draw (108.center) to (110.center);
		\draw (110.center) to (111.center);
		\draw (111.center) to (109.center);
		\draw (117.center) to (118.center);
		\draw (117.center) to (119.center);
		\draw (119.center) to (120.center);
		\draw (120.center) to (118.center);
		\draw [style=discon] (128.center) to (130.center);
		\draw [style=discon] (130.center) to (129.center);
		\draw [style=discon] (128.center) to (131.center);
		\draw [style=discon] (131.center) to (129.center);
		\draw (107.center) to (124.center);
		\draw (105.center) to (115.center);
		\draw (116.center) to (126.center);
		\draw (134) to (133);
		\draw (80) to (134);
		\draw (133) to (138.center);
		\draw (114.center) to (122.center);
		\draw (123.center) to (137.center);
		\draw (127.center) to (136.center);
		\draw (125.center) to (135.center);
		\draw (79) to (113.center);
		\draw (78) to (104.center);
		\draw (81) to (139);
		\draw (139) to (106.center);
	\end{pgfonlayer}
\end{tikzpicture}}
  \end{equation*}
  Hadamard gates can also become extractable after simplifying a frontier with CNOT gates.
  We treat these Hadamards together with those CNOTs in Case (3).
  
\item Layers of CZ gates are handled by separating their gates in two sets.
  Since all CZs commute with each other, we can easily split them into the group of gates whose
  qubits are not involved in the current GMS and the group in which they are.
  The former, similar to the case of single-qubit gates above,
  is placed in the circuit to the right of the current
  GMS compiled as one GMS using~\Cref{eq:cztogzz}.
  The latter set is also compiled as a single GMS and placed in front of the current GMS
  being built.
  CZ gates have the problem that compiling them as GMS gates require all involved qubits to be
  conjugated by Hadamard gates, which in turn get in the way if we try to add additional entangling
  gates into the global gate they belong to.
  In addition to this, CNOTs and CZs are usually extracted in an interleaved way,
  so it is unlikey that we will have two layers of Hadamards canceling out due to two consecutive
  CZ layers.
  For this reason we do not merge CZs into the current GMS and we rather push these gates past it
  whenever possible, since we have a better chance of adding more CNOTs to the GMS instead.
  An example of how CZs are treated is shown below.
  \begin{equation*}
    \scalebox{1}{\input{peephole2.tikz}}
  \end{equation*}
\item Extracted layers of CNOT gates are managed as follows.
  Recalling that we make no assumptions about their structure, we do the following to compile them as
  GMS gates.
  First, we decide how the CNOTs will look like as a series of $\h,\rx$, and $\xx$ gates.
  In~\Cref{sec:extraction} we discussed how extracting CNOTs is done with the purpose of
  simplifying frontier vertices, which in turn creates an extractable Hadamard gate after removing
  the simplified vertex and advancing the frontier.
  We use this Hadamard gate to cancel out the other Hadamard that appears when compiling
  CNOTs as $\xx$ rotations. For example:
  \begin{equation*}
    \scalebox{1}
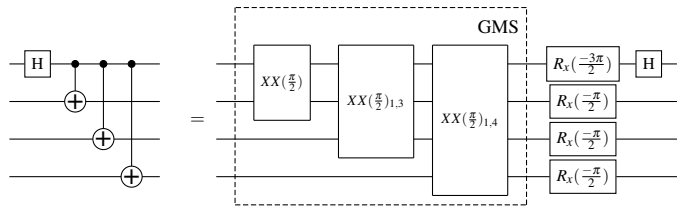
{\begin{tikzpicture}
	\begin{pgfonlayer}{nodelayer}
		\node [style=none] (0) at (-9.5, 1.5) {};
		\node [style=none] (1) at (-9.5, 0.25) {};
		\node [style=none] (2) at (-9.5, -0.75) {};
		\node [style=Gate control] (3) at (-7.75, 1.5) {};
		\node [style=Gate control] (4) at (-7, 1.5) {};
		\node [style=Gate control] (5) at (-6.25, 1.5) {};
		\node [style=borders] (6) at (-7.75, 0.25) {+};
		\node [style=borders] (7) at (-7, -0.75) {+};
		\node [style=borders] (8) at (-6.25, -1.75) {+};
		\node [style=none] (9) at (-9.5, -1.75) {};
		\node [style=none] (10) at (-5.5, 1.5) {};
		\node [style=none] (11) at (-5.5, 0.25) {};
		\node [style=none] (12) at (-5.5, -0.75) {};
		\node [style=none] (13) at (-5.5, -1.75) {};
		\node [style=Quantum gate] (92) at (-8.75, 1.5) {H};
		\node [style=txtnobold] (94) at (-4.5, -0.25) {$=$};
		\node [style=none] (95) at (-4, 1.5) {};
		\node [style=none] (96) at (-4, 0.25) {};
		\node [style=none] (97) at (-4, -0.75) {};
		\node [style=none] (98) at (-4, -1.75) {};
		\node [style=none] (99) at (8.75, 1.5) {};
		\node [style=none] (100) at (8.75, 0.25) {};
		\node [style=none] (101) at (8.75, -0.75) {};
		\node [style=none] (102) at (8.75, -1.75) {};
		\node [style=none] (103) at (-0.25, 2) {};
		\node [style=none] (104) at (1.75, 2) {};
		\node [style=none] (105) at (-0.25, -1.25) {};
		\node [style=none] (106) at (1.75, -1.25) {};
		\node [style=smallest txt] (107) at (0.75, 0.5) {$XX(\frac{\pi}{2})_{1,3}$};
		\node [style=none] (108) at (-3, 2) {};
		\node [style=none] (109) at (-1, 2) {};
		\node [style=none] (110) at (-3, -0.25) {};
		\node [style=none] (111) at (-1, -0.25) {};
		\node [style=smallest txt] (112) at (-2, 1) {$XX(\frac{\pi}{2})_{1,2}$};
		\node [style=none] (116) at (-3, 0.25) {};
		\node [style=none] (117) at (-1, 0.25) {};
		\node [style=none] (118) at (2.25, 2) {};
		\node [style=none] (119) at (4.25, 2) {};
		\node [style=none] (120) at (2.25, -2.25) {};
		\node [style=none] (121) at (4.25, -2.25) {};
		\node [style=smallest txt] (122) at (3.25, 0) {$XX(\frac{\pi}{2})_{1,4}$};
		\node [style=none] (123) at (-3, 1.5) {};
		\node [style=none] (124) at (-1, 1.5) {};
		\node [style=none] (125) at (-0.25, 1.5) {};
		\node [style=none] (126) at (1.75, 1.5) {};
		\node [style=none] (127) at (2.25, 1.5) {};
		\node [style=none] (128) at (4.25, 1.5) {};
		\node [style=none] (129) at (-0.25, 0.25) {};
		\node [style=none] (130) at (1.75, 0.25) {};
		\node [style=none] (131) at (2.25, 0.25) {};
		\node [style=none] (132) at (4.25, 0.25) {};
		\node [style=none] (133) at (-0.25, -0.75) {};
		\node [style=none] (134) at (1.75, -0.75) {};
		\node [style=none] (135) at (2.25, -0.75) {};
		\node [style=none] (136) at (4.25, -0.75) {};
		\node [style=none] (137) at (2.25, -1.75) {};
		\node [style=none] (138) at (4.25, -1.75) {};
		\node [style=Quantum gate] (139) at (6.25, 1.5) {$R_x(\frac{-3\pi}{2})$};
		\node [style=Quantum gate] (140) at (6.25, 0.25) {$R_x(\frac{-\pi}{2})$};
		\node [style=Quantum gate] (141) at (6.25, -0.75) {$R_x(\frac{-\pi}{2})$};
		\node [style=Quantum gate] (142) at (6.25, -1.75) {$R_x(\frac{-\pi}{2})$};
		\node [style=none] (146) at (-3.5, 3) {};
		\node [style=none] (147) at (4.75, 3) {};
		\node [style=none] (148) at (-3.5, -2.5) {};
		\node [style=none] (149) at (4.75, -2.5) {};
		\node [style=txtnobold] (150) at (4, 2.5) {GMS};
		\node [style=Quantum gate] (151) at (8, 1.5) {H};
	\end{pgfonlayer}
	\begin{pgfonlayer}{edgelayer}
		\draw (3) to (6);
		\draw (6) to (1.center);
		\draw (2.center) to (7);
		\draw (7) to (4);
		\draw (4) to (3);
		\draw (4) to (5);
		\draw (5) to (8);
		\draw (8) to (9.center);
		\draw (5) to (10.center);
		\draw (6) to (11.center);
		\draw (7) to (12.center);
		\draw (8) to (13.center);
		\draw (0.center) to (92);
		\draw (92) to (3);
		\draw (103.center) to (104.center);
		\draw (103.center) to (105.center);
		\draw (105.center) to (106.center);
		\draw (106.center) to (104.center);
		\draw (108.center) to (109.center);
		\draw (108.center) to (110.center);
		\draw (110.center) to (111.center);
		\draw (111.center) to (109.center);
		\draw (118.center) to (119.center);
		\draw (118.center) to (120.center);
		\draw (120.center) to (121.center);
		\draw (121.center) to (119.center);
		\draw (124.center) to (125.center);
		\draw (126.center) to (127.center);
		\draw (117.center) to (129.center);
		\draw (130.center) to (131.center);
		\draw (134.center) to (135.center);
		\draw (128.center) to (139);
		\draw (132.center) to (140);
		\draw (136.center) to (141);
		\draw (138.center) to (142);
		\draw [style=discon] (147.center) to (146.center);
		\draw [style=discon] (146.center) to (148.center);
		\draw [style=discon] (148.center) to (149.center);
		\draw [style=discon] (149.center) to (147.center);
		\draw (95.center) to (123.center);
		\draw (96.center) to (116.center);
		\draw (97.center) to (133.center);
		\draw (98.center) to (137.center);
		\draw (142) to (102.center);
		\draw (101.center) to (141);
		\draw (140) to (100.center);
		\draw (139) to (151);
		\draw (151) to (99.center);
	\end{pgfonlayer}
\end{tikzpicture}}
  \end{equation*}
  Removing the Hadamard on the left in turn allows us to potentially add more CNOTs acting on that
  qubit to the GMS.
  
  Knowing how to use Hadamards to compile CNOTs better, we show now how we compile CNOT layers
  into a sequence of $\h,\rx$, and $\xx$ gates by
  deciding on the following for each CNOT in the sequence:
  \begin{enumerate}
  \item A $\cnot_{i,j}$ will have a $\h$ gate to the left in qubit $i$ if
    the previous gate on qubit $i$ was a CNOT from the same layer with $i$ as a target.
    Alternatively, we will also place an $\h$ in the same position if this CNOT is the leftmost CNOT
    in this layer for qubit $i$ and the frontier vertex of qubit $i$ is not extractable.
  \item A $\cnot_{i,j}$ will have a $\h$ gate to the right in qubit $i$ unless the next gate on qubit $i$
    on the right is a CNOT with $i$ as the control, or if there are no more CNOTs
    on the right but there is an adjacent Hadamard on $i$, in which case we cancel them out.
  \end{enumerate}
  Knowing the placement of every Hadamard in the layer, we convert each CNOT into the sequence
  $\rx(-\pi/2)_i\rx(-\pi/2)_j \xx(\pi/2)_{i,j}$ and appropriately push the $\rx$ gates to the right until
  reaching the end of the layer or a Hadamard gate, merging consecutive $\rx$ gates if any.
  Once we have rewritten the CNOT layer, we easily decide which gates can go into the current GMS
  by finding the $\xx$ interactions that can reach the current GMS without clashing with a Hadamard
  (or other extracted gates that are in between them, such as $\rz$ or CZ).
  The gates that we are not able to merge into the GMS are the ones starting the new GMS.
  We give now an example of compiling a CNOT layer with the method above:
  \begin{equation*}
    \scalebox{0.8}{\input{peephole3.tikz}}
  \end{equation*}
\end{enumerate}

We apply the rules above for all incoming gates during circuit extraction until we have fully
extracted our quantum circuit.

\begin{proposition}\label{prop:algorithm}
  The circuit compilation algorithm \texttt{gms_compiler} of~\Cref{sec:algo} compiles arbitrary
  quantum circuits into ones that use Global M{\o}lmer-S{\o}rensen gates as the entangling
  operation.
  The algorithm preserves \emph{gflow} and terminates.
\end{proposition}
\begin{proof}
  The techniques discussed in~\Cref{sec:reduction}
  for turning quantum circuits into simplified graph-like
  ZX-diagrams work for arbitrary quantum circuits and preserve gflow, as shown in
  \cite{backens-there-2021}.
  Similarly, the rewrite rules for extracting $\rz,\h,\cz$, and $\cnot$ gates from
  \Cref{sec:extraction} provably preserve gflow on the ZX-diagram
  \cite{duncan-graph-theoretic-2020}.
  We leverage these same rules during our circuit extraction, with the difference
  being in what is done with the gates after being extracted (on the quantum circuit side), or in
  how we choose which $\cnot$s to extract.
  This allows us to preserve gflow throughout the process and to always be able to find at least
  one extractable vertex in the linear program using~\Cref{prop:commutingExtract}.
  Given that only $\cz$ and $\cnot$ gates are extracted from the diagram, and we compile them into
  GMS gates with the rewrites of~\Cref{sec:peephole}, every entangling operation in the output
  circuit is a GMS gate.
\end{proof}

After finishing the circuit extraction, in practice we end up with a high number of
single-qubit gates due to decomposing CZs and CNOTs into
single- and two-qubit rotations.
To improve the end result, we have also written a simple single-qubit gate reduction
pass that we execute at the very end of our algorithm.
It performs simplifications such as $\h\rx\h=\rz$ and $\h\rz\h=\rx$ and commutes $\rx$ gates through
$\xx$ gates in order to merge as many $\rx$ and $\rz$ gates as we can,
while ensuring that the arrangement of the GMS gates stays unchanged.
In practice, we have observed that this optimization pass significantly reduces the single-qubit
gate count of our end results.

\section{Evaluation Results}
In this section we discuss the implementation of our methods and experimental results when compiling
a variety of quantum circuits, with other algorithms used for comparison.
\subsection{Experimental setup}\mbox{}\\
We run experiments on
QASMBench~\cite{li2022qasmbenchlowlevelqasmbenchmark},
MQT Bench~\cite{quetschlich2023mqtbench}, and
quantum chemistry~\cite{cowtan-phase-2020}
benchmarking circuits.
We compare the performance of our peephole optimization algorithm
using both the linear program and the Patel et al. algorithms.
Furthermore, both implementations are compared to the Qiskit~\cite{qiskit2024}
circuit optimizer and a naive compilation approach.
We use Qiskit as an end-to-end compilation tool for comparison given that,
to the best of our knowledge, there is no publicly available
end-to-end implemented compilation
algorithm from arbitrary circuits to circuits using global rotations, 
We set in Qiskit the optimization level to $3$ (the highest available)
and perform a transpiler pass that uses the $\xx$ gate as the entangling gate,
as Qiskit does not offer transpilation for GMS gates.
The naive algorithm consists of taking the original circuit and greedily creating GMS gates only for
parallel CNOT gates (ones that do not share any qubits), without doing any type of gate commutations.

We have implemented the \lp of~\Cref{sec:lp}
to extract layers of CNOTs and the peephole optimization
logic of~\Cref{sec:peephole} and integrated them into
PyZX~\cite{pyzx},
a Python library used for working with the ZX-calculus.%
\footnote{Code implementation and an extended list of benchmarking results available on Github~\cite{Villoria-Optimization-and-Synthesis}}
We implemented the linear program in the Python-MIP
library~\cite{pythonMIP} using Gurobi~\cite{gurobi}
as backend.
PyZX currently comes with an implementation of the diagram simplification
of~\Cref{sec:reduction} and the base
circuit extraction algorithm explained
in~\Cref{sec:extraction},
where the logic for simplifying the connectivity
of a frontier with CNOT gates utilizes the synthesis algorithm for linear reversible circuits
from Patel et al.~\cite{patel-optimal-2008}.
Experiments were run on a consumer laptop with a timeout set for $30$ minutes for each
circuit, where the ZX approaches were given $15$ minutes for diagram simplification
and $15$ minutes for circuit extraction.
Lastly, the output circuits of our algorithms have been verified by doing an equivalence
check with PyZX against the original circuit (except for some circuits where the
equivalence checker exceeded the five minute timeout).

\begin{table}
  \centering
  \resizebox{\columnwidth}{!}{%
    \begin{tabular}{|l|ccc|ccc|ccc|ccc|}
      \hline
      ~ & ~ & \textbf{Naive Algorithm} & ~ & ~ & \textbf{Peephole +~\cite{patel-optimal-2008} } & ~ & ~ & \textbf{Peephole + LP} & ~ & ~ & \textbf{Qiskit} & ~\rule{0pt}{2.6ex} \\
      Circuit & SQG & Entangling & T & SQG & Entangling & T & SQG & Entangling & T & SQG & Entangling & T \\ \hline\hline
      adder\_n64 & 3431 & 399 & 451.7 & 1963 & 198 & 195.4 & 1977 & 152 & \textbf{ 154.5 } & 1654 & 455 & 411.6 \\
      ae\_n25 & 4865 & 600 & 772 & 1790 & 197 & 199.7 & 2413 & 155 & \textbf{ 160.8 } & 1774 & 558 & 477.7 \\
      basis\_change\_n3 & 135 & 10 & 16 & 109 & 24 & 23.8 & 153 & 29 & 29.6 & 63 & 10 & \textbf{ 10.1 } \\
      basis\_test\_n4 & 343 & 45 & 52.6 & 68 & 7 & \textbf{ 7.8 } & 103 & 11 & 11.8 & 120 & 20 & 19.9 \\
      basis\_trotter\_n4 & 5263 & 581 & 737.7 & 505 & 103 & \textbf{ 104.6 } & 886 & 130 & 133.2 & 1370 & 233 & 233 \\
      bell\_n4 & 88 & 7 & 10.1 & 28 & 5 & \textbf{ 5.2 } & 28 & 5 & \textbf{ 5.2 } & 30 & 5 & \textbf{ 5.2 } \\
      bigadder\_n18 & 980 & 129 & 145.6 & 576 & 77 & 75.6 & 576 & 69 & \textbf{ 70 } & 477 & 130 & 117.6 \\
      bv\_n280 & 2589 & 152 & 153.2 & 165 & 3 & \textbf{ 3.1 } & 165 & 3 & \textbf{ 3.1 } & 307 & 152 & 119.2 \\
      cat\_n260 & 1557 & 259 & 260.1 & 1552 & 259 & 259.5 & 1552 & 259 & 259.5 & 1036 & 259 & \textbf{ 231.1 } \\
      dj\_n25 & 314 & 24 & 33.8 & 38 & 3 & \textbf{ 3.1 } & 38 & 3 & \textbf{ 3.1 } & 84 & 24 & 20.3 \\
      dnn\_n51 & 3338 & 392 & 484.7 & 2619 & 258 & 255.5 & 2363 & 136 & \textbf{ 140.6 } & 1264 & 271 & 255.3 \\
      error\_correctiond3\_n5 & 482 & 49 & 65 & 30 & 4 & \textbf{ 4 } & 30 & 4 & \textbf{ 4 } & 189 & 35 & 34.9 \\
      fredkin\_n3 & 62 & 8 & 9.4 & 45 & 8 & 8.1 & 45 & 8 & 8.1 & 34 & 8 & \textbf{ 7.2 } \\
      ghz\_state\_n255 & 1527 & 254 & 255.1 & 1522 & 254 & 254.5 & 1522 & 254 & 254.5 & 1016 & 254 & \textbf{ 226.7 } \\
      grover\_n2 & 25 & 2 & 3.4 & 14 & 1 & \textbf{ 1.4 } & 14 & 1 & \textbf{ 1.4 } & 14 & 2 & 2.1 \\
      grover-v-chain\_n11 & 23990 & 3280 & 3716.2 & 18690 & 2396 & 2389.7 & 13781 & 1337 & \textbf{ 1347.5 } & 11072 & 3280 & 2931.5 \\
      H2\_UCCSD\_BK\_sto3g & 283 & 38 & 44.7 & 93 & 15 & 15.4 & 92 & 13 & \textbf{ 13.7 } & 108 & 26 & 24.6 \\
      H2\_UCCSD\_JW\_631g & 5318 & 768 & 824.5 & 832 & 113 & 113.6 & 767 & 81 & 83.8 & 277 & 62 & \textbf{ 59.7 } \\
      H2\_UCCSD\_JW\_sto3g & 430 & 56 & 63.6 & 93 & 14 & 14.5 & 93 & 14 & 14.5 & 46 & 9 & \textbf{ 8.8 } \\
      H2\_UCCSD\_P\_sto3g & 282 & 38 & 44.8 & 93 & 15 & 15.5 & 86 & 14 & \textbf{ 14.6 } & 113 & 25 & 23.8 \\
      hhl\_n7 & 1970 & 196 & 246.8 & 1099 & 234 & 248.2 & 1442 & 274 & 289 & 382 & 92 & \textbf{ 85.6 } \\
      hs4\_n4 & 48 & 4 & 6.3 & 22 & 1 & \textbf{ 1.3 } & 22 & 1 & \textbf{ 1.3 } & 28 & 4 & 4 \\
      ising\_n98 & 1651 & 194 & 232.1 & 969 & 194 & \textbf{ 183.8 } & 969 & 194 & \textbf{ 183.8 } & 1041 & 194 & 189.3 \\
      iswap\_n2 & 24 & 2 & 3.3 & 12 & 1 & \textbf{ 1.4 } & 12 & 1 & \textbf{ 1.4 } & 14 & 2 & 2.1 \\
      knn\_n67 & 2250 & 264 & 313.8 & 1384 & 48 & 46.8 & 1559 & 23 & \textbf{ 23.6 } & 924 & 231 & 209.9 \\
      LiH\_UCCSD\_BK\_sto3g & 60162 & 8680 & 9068.8 & 5998 & 848 & 842.2 & 5231 & 484 & 502 & 1505 & 414 & \textbf{ 380.9 } \\
      LiH\_UCCSD\_JW\_sto3g & 54020 & 8064 & 8500.4 & 6290 & 818 & 811.2 & 5449 & 502 & 519.2 & 1533 & 388 & \textbf{ 364 } \\
      LiH\_UCCSD\_P\_sto3g & 53330 & 7640 & 8029.7 & 6621 & 912 & 904.5 & 5522 & 481 & \textbf{ 498.7 } & 12217 & 3681 & 3328 \\
      linearsolver\_n3 & 46 & 4 & 5.9 & 18 & 2 & \textbf{ 2.2 } & 18 & 2 & \textbf{ 2.2 } & 19 & 4 & 4.1 \\
      lpn\_n5 & 27 & 2 & 2.9 & 15 & 3 & 3 & 15 & 3 & 3 & 5 & 2 & \textbf{ 1.8 } \\
      multiplier\_n15 & 1670 & 222 & 248.4 & 741 & 91 & 90.5 & 613 & 50 & \textbf{ 51.3 } & 811 & 222 & 201 \\
      multiply\_n13 & 316 & 39 & 43.6 & 202 & 16 & \textbf{ 15.7 } & 202 & 16 & \textbf{ 15.7 } & 173 & 40 & 37 \\
      pea\_n5 & 318 & 42 & 50.9 & 72 & 13 & 12.9 & 72 & 12 & \textbf{ 12 } & 84 & 17 & 16.6 \\
      portfolioqaoa\_n17 & 5661 & 788 & 912.3 & 2502 & 171 & 169.8 & 2317 & 134 & \textbf{ 133.6 } & 2637 & 816 & 728 \\
      portfoliovqe\_n13 & 3068 & 234 & 349.5 & 453 & 48 & 47.2 & 566 & 35 & \textbf{ 36.5 } & 555 & 234 & 190.6 \\
      qaoa\_n16 & 624 & 62 & 83.8 & 242 & 31 & \textbf{ 30.5 } & 242 & 31 & \textbf{ 30.5 } & 378 & 64 & 64 \\
      qec9xz\_n17 & 255 & 30 & 39.3 & 64 & 5 & \textbf{ 5.2 } & 64 & 5 & \textbf{ 5.2 } & 75 & 32 & 26.9 \\
      qec\_en\_n5 & 96 & 10 & 12.2 & 30 & 4 & \textbf{ 4.1 } & 30 & 4 & \textbf{ 4.1 } & 53 & 10 & 9.9 \\
      qf21\_n15 & 934 & 115 & 142.4 & 415 & 46 & 45.7 & 400 & 37 & \textbf{ 37.8 } & 480 & 115 & 106.8 \\
      qft\_n29 & 6177 & 812 & 1037.3 & 2219 & 235 & 238.1 & 3286 & 202 & \textbf{ 210.4 } & 2055 & 652 & 586.5 \\
      qftentangled\_n25 & 4938 & 638 & 714.2 & 1832 & 237 & 239 & 2522 & 175 & \textbf{ 182.3 } & 2029 & 604 & 547.9 \\
      qpe\_n9 & 368 & 43 & 53.1 & 150 & 21 & \textbf{ 20.4 } & 150 & 21 & 20.5 & 201 & 43 & 40.7 \\
      qpeexact\_n25 & 4801 & 628 & 791.4 & 2098 & 213 & 213.8 & 2661 & 157 & \textbf{ 163 } & 1950 & 585 & 529.6 \\
      qpeinexact\_n25 & 4861 & 636 & 800.3 & 2215 & 225 & 225.1 & 2644 & 168 & \textbf{ 174.6 } & 2005 & 594 & 539.3 \\
      qram\_n20 & 1026 & 129 & 143.4 & 746 & 88 & 86.6 & 653 & 47 & \textbf{ 48.2 } & 498 & 130 & 117.3 \\
      qrng\_n4 & 12 & 0 & 0.3 & 8 & 0 & \textbf{ 0.2 } & 8 & 0 & \textbf{ 0.2 } & 12 & 0 & 0.3 \\
      qugan\_n39 & 2516 & 295 & 358.4 & 1858 & 169 & 166.5 & 1761 & 108 & \textbf{ 110.1 } & 947 & 205 & 193.1 \\
      qwalk-noancilla\_n8 & 30879 & 4302 & 4849.7 & 14310 & 2195 & 2161.4 & 16233 & 1909 & \textbf{ 1948.1 } & 14509 & 4290 & 3848.7 \\
      qwalk-v-chain\_n19 & 14847 & 1626 & 1898.1 & 12503 & 2043 & 2027.2 & 12216 & 1896 & 1920 & 5160 & 1614 & \textbf{ 1450.9 } \\
      random\_n23 & 9934 & 1069 & 1336.3 & 7275 & 807 & 796.4 & 6982 & 429 & \textbf{ 437.8 } & 4533 & 1032 & 976.4 \\
      realamprandom\_n21 & 4200 & 594 & 633.6 & 847 & 90 & 87 & 502 & 41 & \textbf{ 42.7 } & 1445 & 630 & 507.9 \\
      sat\_n11 & 1991 & 243 & 275.3 & 834 & 135 & \textbf{ 134.9 } & 896 & 135 & 137.5 & 1007 & 252 & 233.8 \\
      seca\_n11 & 702 & 79 & 93.5 & 293 & 36 & 35.9 & 252 & 28 & \textbf{ 28.6 } & 278 & 80 & 71.6 \\
      simon\_n6 & 126 & 14 & 16.8 & 5 & 1 & \textbf{ 0.8 } & 5 & 1 & \textbf{ 0.8 } & 54 & 14 & 12.8 \\
      su2random\_n25 & 6000 & 856 & 914.4 & 467 & 31 & \textbf{ 31.1 } & 467 & 31 & \textbf{ 31.1 } & 2094 & 900 & 730.1 \\
      swap\_test\_n83 & 2876 & 328 & 395.1 & 1441 & 38 & 37.5 & 2065 & 23 & \textbf{ 23.5 } & 1189 & 287 & 260.7 \\
      teleportation\_n3 & 22 & 2 & 2.9 & 12 & 2 & 2.2 & 12 & 2 & 2.2 & 8 & 2 & \textbf{ 1.9 } \\
      toffoli\_n3 & 51 & 6 & 7 & 41 & 6 & 6.2 & 37 & 5 & \textbf{ 5.3 } & 31 & 6 & 5.8 \\
      twolocalrandom\_n21 & 4200 & 594 & 633.6 & 847 & 90 & 87 & 502 & 41 & \textbf{ 42.7 } & 1445 & 630 & 507.9 \\
      variational\_n4 & 138 & 16 & 20.4 & 86 & 14 & 14.5 & 75 & 12 & 12.2 & 51 & 8 & \textbf{ 8.1 } \\
      vqe\_n16 & 420 & 30 & 31.9 & 471 & 48 & 47.4 & 307 & 21 & \textbf{ 21.8 } & 172 & 30 & 27.3 \\
      vqe\_uccsd\_n8 & 36096 & 5284 & 5772.7 & 2542 & 389 & 389.8 & 2521 & 329 & \textbf{ 334.6 } & 16623 & 4807 & 4394.3 \\
      wstate\_n380 & 6823 & 757 & 883.6 & 3797 & 759 & 719.4 & 3797 & 759 & 719.4 & 3415 & 758 & \textbf{ 676.5 } \\
      \hline
    \end{tabular}%
  }
  \caption{Numerical results of the compilation of a variety of quantum circuits from the QASMBench,
    MQTBench, and quantum chemistry benchmarking suites.
    For each compilation method, we showcase the final single-qubit gate count (SQG column),
    entangling gate count (Entangling column) and circuit execution time in milliseconds
    (T column).
    For \textit{Qiskit}, the entangling gates are $\xx$ gates,
    while for our methods and the naive algorithm they refer to GMS gates.
    Results in bold point to the method with the fastest circuit execution time.
    \label{table:results}}
\end{table}

\subsection{Results and discussion}\mbox{}\\
Table (\ref{table:results}) shows benchmarking results for a variety of circuits
from the three benchmark suites.
We display the single-qubit gate counts (SQG), two-qubit or global gate counts (Entangling), and the
execution time of the circuit (T) in milliseconds.
The best performance is marked in bold.
Data regarding the gate times for each operation are taken from the latest benchmarks of the
IonQ Forte~\cite{Chen2024benchmarkingtrapped}
quantum computer, which is of $110\mu s$ for single-qubit gates and $672\mu s$ for entangling
gates.
Note that individual $\xx$ gates and $\gms$ gates take the same time to execute.
We have included all circuits from the benchmarking datasets, excluding only the ones
with a lower qubit count from each circuit family and removing repeated circuits across
benchmark suites.
Circuit names have appended their qubit count in the form of \textit{name\_nX} where X is
the qubit count of the circuit.
The section labeled \textit{Naive Algorithm} corresponds to the naive compilation that would be needed
to execute the original circuit in the hardware with only adding parallel CNOTs into GMS gates.
The \textit{Peephole +}~\cite{patel-optimal-2008} and \textit{Peephole + LP}
sections correspond to the performance of our algorithms with
the peephole optimization using the Patel algorithm and the LP for the frontier simplification,
respectively.

We can see that for most circuits the peephole + LP and peephole + Patel
implementations yield the faster circuits,
with some exceptions that we will mention shortly.
The peephole optimization alone rarely outperforms the
implementation with the \lp, although in some cases such as the circuits \texttt{basis_test} and
\texttt{basis_trotter} it is able to reach a lower amount of GMS gates.
Although rare, this can happen given that the LP only takes into account the information on the
current frontier vertices and not the whole ZX-diagram, so different choices of $\cnot$s to extract
may yield a simpler diagram later on in the extraction.
The drawback however is that the implementation with the LP program takes the longest to compile
the quantum circuits, followed by the peephole optimization and then Qiskit.
Notice also that in most cases the peephole optimization alone is sufficient to outperform Qiskit.
If we inspect the circuits in which Qiskit performs better, we can see that this is due to the
original structure of the circuit not being very amenable for grouping commuting CNOTs or CZs
into global gates.
For example, the \texttt{cat, ghz}, and \texttt{wstate} circuits are built almost
exclusively with one ``CNOT ladder'', which is a sequence of CNOTs where the control of a CNOT
is the target of the next (or vice-versa).
Some of these circuits could have equivalent implementations with a
commuting CNOT layer instead of a ladder, but when given the less convenient implementation
our algorithm cannot find the alternative.
This shows that even when using optimization algorithms, we still require sensible circuit
design in the first place, though having an algorithm for compiling arbitrary circuits
that can be integrated in automated tools has significant value.
Similarly, \texttt{fredkin} and \texttt{lpn} are also rigid circuits that are constructed with
CNOTs that are not amenable to be put together into a GMS.
Lastly, looking at the entangling gate counts, we can see how both of our algorithms are able
to reduce their amount significantly in comparison to both the original implementation and
Qiskit in most of the cases.
This is particularly interesting given that the ZX-calculus optimization routines usually
excel at reducing non-Clifford phase gates and struggle with reducing two-qubit gate counts,
sometimes even increasing them~\cite{staudacher-reducing-2023}.

\section{Conclusions and Future Work}
In this work we have presented a circuit synthesis and optimization algorithm for
ion trap quantum devices.
It is based on the ZX-calculus and consists of a circuit extraction routine that has as
objective to group entangling gates into global gates, a special entangling operation
available in ion trap devices that performs two-qubit gates in parallel.
The algorithm consists of a series of peephole optimization style rewrite rules that for the most
part puts layers of two-qubit gates next to each other and removes Hadamard gates
so more entangling gates can be compiled together.
We also develop a Linear Program that decides the shape of extracted CNOT layers in a
way that these can be compiled with only one global gate.
We ran bechmarks on several quantum circuits against a naive implementation and Qiskit
and report that in most of the cases both of our implementations outperform the others.

As future work, we remark that variations on the standard circuit extraction algorithm
is a largely unexplored topic.
This and a related work from Staudacher et al.~\cite{staudacher-multi-controlled-2024}
show how changing the way in which a
circuit is extracted from a ZX-diagram can result in new compilation algorithms
for specific hardware platforms.
Modifying circuit extraction is not an easy task however, as one has to ensure that the
gflow of the diagram is preserved.
In our case, we circumvent this issue by using existing rules and changing how they are
used rather than coming up with new ones.
Another interesting direction would be to investigate if allowing different coupling
strengths for each qubit pair (as in the EASE gate) would simplify our circuits even further,
given that for now we only assume the ability of turning arbitrary interactions off
but all active interactions must have the same rotation angle.
An algorithm that e.g. extracts two-qubit \emph{phase gadgets}~\cite{cowtan-phase-2020}
instead of CNOTs could group these together into EASE gates.

\section*{Acknowledgements}
This work was funded by the European Union under Grant Agreement 101080142, EQUALITY project.
\bibliographystyle{eptcs}
\bibliography{paper}
\appendix
\section{Graph-theoretic Definitions}\label[appendix]{app:graph}

Here we present some technical definitions related to the graph-theoretic view of ZX-diagrams.

\begin{definition}[\cite{duncan-graph-theoretic-2020}, Definition 3.1.]\label{def:graph-like}
  A ZX-diagram is in graph-like form if it fulfills the following properties:
  \begin{itemize}
  \item All spiders are Z-spiders.
  \item Z-spiders are connected only by Hadamard edges.
  \item There are no parallel edges between spiders and no edges from a spider to itself.
  \item All inputs and outputs of the diagram are connected to a Z-spider via a non-Hadamard edge.
  \item A Z-spider is connected to at most one input or output.
  \end{itemize}
\end{definition}

\begin{definition}[Adapted from \cite{backens-there-2021}]]\label{def:labelled}
  A labelled open graph is a tuple $\Gamma = (G,I,O,\lambda)$ such that:
  \begin{enumerate}
  \item $G = (V,E)$ is an undirected (simple) graph.
  \item $I,O\subseteq V$ are subsets of the vertices representing input and output vertices,
    respectively.
    We denote the non-input and non-output vertices by
    $\comp{I}:= V\setminus I$ and $\comp O:= V\setminus O$.
    The intersection between $I$ and $O$ is not necessarily empty.
  \item $\lambda: \comp{O} \to \{XY,YZ,XZ\}$ is a labeling of the non-output vertices
    into measurement planes.
    The assignment of measurement plane for a non-output vertex $v$ indicates whether a spider
    \begin{tikzpicture}
	\begin{pgfonlayer}{nodelayer}
		\node [style=none] (0) at (-1, 0) {};
		\node [style=Green Node] (1) at (0, 0) {$\alpha$};
	\end{pgfonlayer}
	\begin{pgfonlayer}{edgelayer}
		\draw [style=edge] (0.center) to (1);
	\end{pgfonlayer}
\end{tikzpicture}
, \begin{tikzpicture}
	\begin{pgfonlayer}{nodelayer}
		\node [style=none] (0) at (-0.5, 0) {};
		\node [style=Green Node] (1) at (0.5, 0) {$\alpha$};
	\end{pgfonlayer}
	\begin{pgfonlayer}{edgelayer}
		\draw [style=hadamard edge] (0.center) to (1);
	\end{pgfonlayer}
\end{tikzpicture}
 or \begin{tikzpicture}
	\begin{pgfonlayer}{nodelayer}
		\node [style=none] (0) at (-1, 0) {};
		\node [style=Green Node] (1) at (1.25, 0) {$\alpha$};
		\node [style=Green Node] (2) at (0, 0) {$\frac{\pi}{2}$};
	\end{pgfonlayer}
	\begin{pgfonlayer}{edgelayer}
		\draw [style=hadamard edge] (2) to (1);
		\draw [style=edge] (0.center) to (2);
	\end{pgfonlayer}
\end{tikzpicture}
 is attached to $v$
    when $\lambda(v)$ is equal to $XY,YZ,$ or $XZ$ (respectively),
    for some angle $\alpha$.
  \end{enumerate}
  For a vertex $v \in V$ we denote the set of neighbours of $v$ by $N(v)$.
  For any $K\subseteq V$, we define the \emph{odd neighbourhood}
  of $K$ as $\odd{K}:=\{w\in V\ |\ |N(w) \cap K|\equiv 1 \pmod{2}\}$, which is the set of vertices
  that have an odd number of neighbours in $K$.
\end{definition}
We can use labelled open graphs $(G,I,O,\lambda)$ to better reason about
a ZX-diagram that has been turned into graph-like form.
The underlying graph $G$ models the Z-spiders and the Hadamard edges between them,
the input/output vertices represent where the quantum circuit begins and ends,
$\lambda$ encodes the measurement planes, and an additional labeling
$\alpha: \comp{O} \to [0,2\pi)$ the measurement angles.
Our ZX-diagrams stay in graph-like form during their simplification and during
extraction, so we are able to reason about them and their flow properties in this way
throughout the whole process.

\begin{definition}[\cite{backens-there-2021}, Definition 2.36]\label{def:gflow}
  A generalized flow (or \textit{gflow}) on a labelled open graph $(G,I,O,\lambda)$
  is a tuple $(g,\prec)$ consisting of a map
  $g: \comp{O} \to 2^{\comp{I}}$ (where $ 2^{\comp{I}}$ is the powerset of $\comp{I}$) and a partial
  order $\prec$ in $V$ such that for all
  $v\in \comp{O}$:
  \begin{itemize}
  \item If $w\in g(v)$ and $w\ne v$, then $v\prec w$.
  \item If $w\in \odd{g(v)}$ and $v\ne w$ then $v\prec w$.
  \item If $\lambda(v) =$ XY then $v \notin g(v)$ and $v\in \odd{g(v)}$.
  \item If $\lambda(v) =$ XZ then $v \in g(v)$ and $v\in \odd{g(v)}$.
  \item If $\lambda(v) =$ YZ then $v \in g(v)$ and $v\notin \odd{g(v)}$.
  \end{itemize}
  We call $g(v)$ the \textit{correction set} of $v$, and it captures how an unwanted measurement
  outcome in $v$ could be accounted for by adaptively correcting later measurements.
\end{definition}
The existence of gflow in a labelled open graph can be understood as the ability to
correct the possible measurement errors that might occur in the computation when
interpreting graph-like diagrams as measurement patterns.

\clearpage
\section{Linear Program}\label[appendix]{app:lp}

We present here the \lp of~\Cref{sec:lp} with all its constraints linearized.
The input is an $n\times m$ adjacency matrix $M$, and the set of variables consists
of $x_{i,j}, y_i, z_i, k_{i,\ell}, t_{i,j}, G_{i,\ell},c$.
The following constraints hold for all $i,\ell,p\in\{1,\dots,n\}, j\in\{1,\dots,m\}$.
\begin{align*}
  \text{maximize} & \quad \sum_{i=1}^n n z_i \ \ - c&\\
  \text{s.t.} & \quad \sum_{i=1}^n z_i  &&\geq  1 && \quad (1)\\
  \\
                  & \quad \sum_{j=1}^m x_{ij} \hspace{0.25em} -1  &&\geq  y_i  && \quad (2) \\
                  & \quad \sum_{j=1}^m x_{ij} \hspace{0.25em} -1  &&\leq  m\ y_i  && \quad (3) \\
                  & \quad y_i +z_i &&= 1 && \quad (4) \\
  \\
                  & \quad t_{i,j} && \geq 0  && \quad (5)\\
                  & \quad t_{i,j} && \leq \lfloor n/2\rfloor  && \quad (6)\\
                  & \quad x_{i,j} && \geq 0 && \quad (7)\\
                  & \quad x_{i,j} && \leq 1 && \quad (8)\\
                  & \quad x_{i,j} &&= \sum_{\ell=1}^n G_{i,\ell}M_{\ell,j} -2t_{i,j}  && \quad (9) \\
  \\
                  & \quad G_{i,i} &&= 1 && \quad (10)\\
                  & \quad k_{i,\ell} &&\le 1 - G_{\ell,p} && \quad (11)\\
                  & \quad k_{i,\ell} &&\ge \sum_{p\ne\ell}\left[1-G_{\ell,p}\right] - (m-2)  && \quad (12)\\
                  & \quad (1-G_{i,\ell}) + k_{i,\ell} && \ge 1  && \quad (13)\\
  \\
                  & \quad c &&= \sum_{i\ne \ell} G_{i,\ell}   && \quad (14)
\end{align*}

We have that the binary variable $z_i=1$ if frontier vertex $i$ is extractable (its row
has a single $1$).
Constraint $(1)$ ensures that we will have at least one simplified frontier vertex.
Contraints $(2-4)$ uses auxiliary binary variables $y_i$ to encode $y_i = 0$ iff
$\sum_{j=1}^m x_{i,j} = 1$ and $y_i$ OR $z_i$.
These conditions are equivalent to $z_i=1$ if row $i$ of $X=GM$ has a single $1$.
Constraints $(5-9)$ use auxiliary integer variables $t_{i,j}$ to linearize the
matrix multiplication $X=GM$, necessary since we are working with binary matrices and
the XOR operation must be linearized.
The constraints $(10-13)$ are used to encode constraints four and five of the \lp in~\Cref{sec:lp}
using the auxiliary binary variables $k_{i,\ell}$ to represent the product part of constraint five.
The last constraint $(13)$ uses integer variable $c$ to count off diagonal entries of $G$.
This gives a linear program with $2nm + 2n^2 + 2n$ variables and $5nm + n^3 + 3n + 1$ constraints.

\section{Pseudocode}\label[appendix]{app:pseudocode}
\begin{algorithm}[H]
\caption{GMS compiler}
\begin{algorithmic}

\Function{extract_sqgs}{frontier, $O$}
\State $\triangleright$ Get extractable single-qubit gates from frontier (\Cref{sec:extraction})
\State sqgs $\gets$ \textproc{get_sqgs}(frontier)
\ForAll {gate $\in$ sqgs}
        \State $O \gets$ place gate after latest GMS if it commutes, place before otherwise
\EndFor
\EndFunction

\Function{extract_czs}{frontier, $O$}
\State $\triangleright$ Get extractable CZs from frontier (\Cref{sec:extraction})
\State czs $\gets$ \textproc{get_czs}(frontier)
\State czs_left, czs_right $\gets$ \textproc{classify_czs}(czs)
        \Comment{(\Cref{sec:peephole}, Case (2)) }
\State $O \gets$ place czs_right after latest GMS
\State $O \gets$ place czs_left before latest GMS
\EndFunction

\Function{extract_cnots}{cnots, $O$}

\ForAll{$\cnot$ $\in$ cnots}
    \State  $\triangleright$ Compile $\cnot$ as sequence of $\h,\xx,\rx$ gates
    \State h_left, $\xx$, rxs, h_right $\gets$ \textproc{classify_cnot}($\cnot$)
    \Comment{(\Cref{sec:peephole}, Case (3)) }
    \State $O \gets$ $\rx\in$ rxs and h_right are commuted to the right
    \If {$\xx\in$ compiled_cnot fits in latest GMS}
        \State  $O \gets$ add $\xx$ to latest GMS
    \Else
        \State $O \gets$ create new GMS with $\xx$ on it
    \EndIf
    \State $O \gets$ place h_left to the left of the $\xx$ gate
\EndFor
\EndFunction
\vspace{1em}

\Function{gms_compiler}{$I$}\Comment{Input: quantum circuit $I$}
\State zx_diagram $\gets$ \textproc{circuit_to_ZX}($I$)
        \Comment{Circuit to graph-like diagram (\Cref{sec:reduction}})
\State zx_diagram $\gets$ \textproc{reduce_diagram}(zx_diagram)
        \Comment{Simplify diagram (\Cref{sec:reduction})}
\State $O$ $\gets \varnothing$ \Comment{Initialize empty output circuit} 
\While{frontier $\in$ zx_diagram}
\State \textproc{extract_sqgs}(frontier, $O$)
        \Comment{Compile single-qubit gates (\Cref{sec:peephole}, Case (1))}
\State \textproc{extract_CZs}(frontier, $O$)
        \Comment{Compile CZs (\Cref{sec:peephole}, Case (2))}
\State cnots $\gets$ \textproc{LP}(frontier)
\Comment{LP to simplify frontier with CNOTs (\Cref{sec:lp})}
\State \textproc{extract_cnots}(cnots, $O$)
        \Comment{Compile CNOT layer (\Cref{sec:peephole}, Case (3))}
\EndWhile

\State $O \gets$ \textproc{simplify_sqgs}($O$)
\Comment{Merge rotation gates and reduce Hadamards(\Cref{sec:peephole})}
\State \textbf{return} $O$
\EndFunction
\vspace{1em}
\end{algorithmic}
\end{algorithm}

\end{document}